\documentclass{article}

\pdfoutput=1

\usepackage{amssymb,amsmath,amsfonts}
\usepackage{graphicx}
\usepackage{cite}

\usepackage{amsthm}

\newtheorem{theorem}{Theorem}[section]
\newtheorem{proposition}[theorem]{Proposition}
\newtheorem{corollary}[theorem]{Corollary}
\newtheorem{lemma}[theorem]{Lemma}

\newcommand{\ep}{\varepsilon}

\begin{document}

\title{\bf Droplet phases in non-local Ginzburg-Landau models with
  Coulomb repulsion in two dimensions}


\author{Cyrill B. Muratov
\thanks{ Department of Mathematical Sciences, New Jersey Institute
    of Technology, Newark, NJ 07102, USA} 
}

\date{\today}

\numberwithin{equation}{section}

\maketitle

\begin{abstract}
  We establish the behavior of the energy of minimizers of non-local
  Ginzburg-Landau energies with Coulomb repulsion in two space
  dimensions near the onset of multi-droplet patterns. Under suitable
  scaling of the background charge density with vanishing surface
  tension the non-local Ginzburg-Landau energy becomes asymptotically
  equivalent to a sharp interface energy with screened Coulomb
  interaction. Near the onset the minimizers of the sharp interface
  energy consist of nearly identical circular droplets of small size
  separated by large distances. In the limit the droplets become
  uniformly distributed throughout the domain. The precise asymptotic
  limits of the bifurcation threshold, the minimal energy, the droplet
  radii, and the droplet density are obtained.
\end{abstract}

\section{Introduction}
\label{sec:intro}

Spatial patterns are often a result of the competition between
thermodynamic forces operating on different length scales. When
short-range attractive interactions are present in a system, phase
separation phenomena can be observed, resulting in aggregation of
particles or formation of droplets of new phase, which evolve into
macroscopically large domains via coarsening or nucleation and growth
(see e.g. \cite{bray94}). This process, however, can be frustrated in
the presence of long-range repulsive forces. As the droplets grow, the
contribution of the long-range interaction may overcome the
short-range forces, whereby suppressing further growth. This mechanism
was identified in many energy-driven pattern forming systems of
different physical nature, such as various types of ferromagnetic
systems, type-I superconductors, Langmuir layers, multiple polymer
systems, etc., just to name a few
\cite{landau8,grosberg,ko:book,vedmedenko,muthukumar97,desimone00,%
  choksi08,choksi01,seul95,yu07}. Remarkably, these systems often
exhibit very similar pattern formation behaviors
\cite{seul95,kohn07iciam}.

One important class of systems with competing interactions are systems
in which the long-range repulsive forces are of Coulomb type (for an
overview, see \cite{m:phd,m:pre02} and references therein). The nature
of the Coulombic forces may be very different from system to
system. For example, these forces may arise when particles undergoing
phase separation carry net electric charge
\cite{care75,emery93,chen93,nyrkova94}, or they may be a consequence
of entropic effects associated with chain conformations in polymer
systems
\cite{ohta86,bates99,matsen02,degennes79,stillinger83}. Coulomb
interactions may also arise indirectly as a result of
diffusion-mediated processes \cite{ko:book,ohta90,glotzer95}. All this
makes systems with repulsive Coulombic interactions a ubiquitous
example of pattern forming systems.

Studies of systems with competing short-range attractive interactions
and long-range repulsive Coulomb interactions go back to the work of
Ohta and Kawasaki, who proposed a non-local extension of the
Ginzburg-Landau energy in the context of diblock copolymer systems
\cite{ohta86}. Even though its validity for diblock copolymer systems
may be questioned \cite{matsen96,matsen02,choksi03,mnog09}, the
Ohta-Kawasaki model is applicable to a great number of physical
problems of different origin \cite{m:pre02}. On the other hand,
mathematically Ohta-Kawasaki model presents a paradigm of
energy-driven pattern forming systems which has been receiving a
growing degree of attention
\cite{muller93,ren00,choksi01,ren03,ren07jns,ren07rmp,ren06,roger08,alberti09}.

The Ohta-Kawasaki energy is a functional of the form
\cite{ohta86,m:phd,m:pre02,nishiura95,ohta90}:
\begin{eqnarray}
  \label{eq:GL}
  \mathcal E[u] & = & \int_\Omega \left( \frac{\ep^2}{2} |\nabla u|^2
    + W(u) \right) dx \nonumber \\ 
  & + & \frac12 \int_\Omega \int_\Omega (u(x) - \bar u) G_0(x, y) (u(y)
  - \bar u) dx \, dy. 
\end{eqnarray}
Here, $u: \Omega \to \mathbb R$ is a scalar quantity denoting the
``order parameter'' in a bounded domain $\Omega \subset \mathbb
R^d$. Different terms of the energy are as follows: the first term
penalizes spatial variations of $u$ on the scales shorter than $\ep$,
the second term, in which $W$ is a symmetric double-well potential
drives local phase separation towards the minima of $W$ at $u = \pm
1$, and the last term is the long-range interaction, whose Coulombic
nature comes from the fact that the kernel $G_0$ solves the Neumann
problem for
\begin{eqnarray}
  \label{eq:G0}
  -\Delta G_0(x,y) = \delta(x-y) - {1 \over |\Omega|}, \qquad \int_\Omega
  G_0(x,y) dx = 0, 
\end{eqnarray}
where $\Delta$ is the Laplacian in $x$ and $\delta(x)$ is the Dirac
delta-function. The parameter $\bar u$ denotes the prescribed uniform
background charge, and the overall ``charge neutrality'' is ensured
via the constraint
\begin{eqnarray}
  \label{eq:solvab}
  {1 \over |\Omega|} \int u \, dx = \bar u. 
\end{eqnarray}
It is important to note that the kernel $G_0$ solves (\ref{eq:G0}) in
the space of the same dimensionality as the order parameter $u$ (not
to be confused with the case in which the kernel solves the Laplace's
equation in the space of higher spatial dimensionality, as is common
in many other systems with competing interactions, see e.g.
\cite{desimone00,emery93}).

\begin{figure}
  \centering
  \includegraphics[width=8cm]{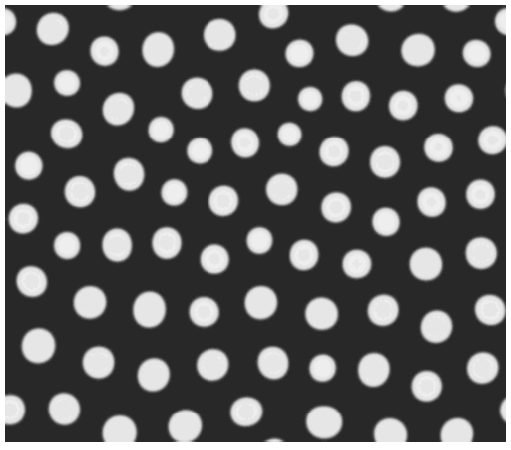}
  \caption{A multi-droplet pattern: density plot of $u$ in a local
    minimizer of $\mathcal E[u]$ with $W(u) = \tfrac14 (1 - u^2)^2$
    obtained numerically for $\bar u = -0.5$, $\ep = 0.025$, and
    $\Omega = [0, 11.5) \times [0, 10)$, with periodic boundary
    conditions. Dark regions correspond to $u \approx -1$, and light
    regions correspond to $u \approx 1$ (from \cite{m:pre02}).}
  \label{fig:pattern}
\end{figure}

The parameter $\ep > 0$ in (\ref{eq:GL}) determines both the scale of
the short-range interaction and the magnitude of the interfacial
energy between the regions with different values of $u$ when $\ep$ is
sufficiently small. In fact, it is known that no patterns can form in
the system if $\ep$ is sufficiently large
\cite{m:phd,m:pre02,choksi09}. On the other hand, when $\ep \ll 1$,
the first term in the functional $\mathcal E$ becomes a singular
perturbation, giving rise to ``domain structures'' (see
Fig.~\ref{fig:pattern}), which are of particular physical
interest. These patterns consist of extended regions in which $u$ is
close to one of the minima of the potential $W$, separated by narrow
domain walls. In this situation one can reduce the energy functional
appearing in (\ref{eq:GL}) to an expression in terms of the interfaces
alone. In \cite{m:phd,m:pre02}, such a reduction was performed for
$\mathcal E$ using formal asymptotic techniques (see also
\cite{petrich94,goldstein96,ren00,roger08}) and leads to the following
reduced energy (for simplicity of notation, we choose the
normalizations in such a way that the parameter $\ep$ is, in fact, the
domain wall energy, see Sec. \ref{sec:conn} for details):
\begin{eqnarray}
  \label{eq:E}
  E[u] = {\ep \over 2} \int_\Omega |\nabla u| \, dx + \frac12
  \int_\Omega \int_\Omega (u(x) - \bar u) G(x, y) (u(y) - \bar u) \,
  dx \, dy.  
\end{eqnarray}
Here the function $u$ takes on values $\pm 1$ throughout $\Omega$, and
the kernel $G$ is the {\em screened} Coulomb kernel, i.e., it solves
the Neumann problem for
\begin{eqnarray}
  \label{eq:Gk}
  -\Delta G(x,y) + \kappa^2 G(x,y) = \delta(x-y), 
\end{eqnarray}
with some $\kappa > 0$.  The constant $\kappa$ has the physical
meaning of the inverse of the Debye screening length
\cite{m:phd,m:pre02}.  Note that the sharp interface energy $E$ with
the unscreened Coulomb kernel (i.e. with $\kappa = 0$) was derived by
Ren and Wei as the $\Gamma$-limit of the diffuse interface energy
$\mathcal E$ under assumptions of weak non-local coupling (i.e., with
an extra factor of $\ep$ in front of the Coulomb kernel) and $\bar u
\in (-1,1)$ independent of $\ep$, as $\ep \to 0$ \cite{ren00} (see
also \cite{roger08}; note that this case is also equivalent to
considering $\mathcal E$ on the domain of size $O(\ep^{1/3})$). At the
same time, screening becomes important near the transition between the
uniform and the patterned states which occurs near $|\bar u| = 1$, the
case of interest in the present paper \cite{m:phd,m:pre02}. Note that
in the presence of screening the neutrality condition in
(\ref{eq:solvab}) is relaxed.

In this paper, we rigorously establish the relation between the sharp
interface energy $E$ and the diffuse interface energy $\mathcal E$,
and analyze the precise behavior of minimizers of the sharp interface
energy $E$ for $\ep \ll 1$ in the vicinity of the transition from the
trivial minimizer to patterned states occurring near $|\bar u| =
1$. We note that despite the apparent simplicity of the expression for
$E$, the minimizers of $E$ exhibit quite an intricate dependence on
the parameters for $\ep \ll 1$ and $|\bar u| \simeq 1$. Our analysis
in this paper will be restricted to the case $d = 2$. While a number
of our results can be extended to arbitrary space dimensions, our
methods to obtain sharp estimates for the energy of minimizers rely
critically on the properties of minimal curves in two dimensions and
the logarithmic behavior of the Green's function of the
two-dimensional Laplacian near the singularity. Therefore, they cannot
be readily extended to other spatial dimensionalities, and, indeed,
one would expect certain important differences between these cases and
the case of two space dimensions. At the same time, we will show that
in the case $d = 2$ it is possible to obtain rather detailed
information about the structure of the transition near $|\bar u| = 1$
in terms of energy. Let us note that, since the case $d = 1$ is now
well-understood \cite{muller93,ren00,ren03,yip06}, the remaining open
case of physical interest is that of $d = 3$.

Before turning to the analysis, let us briefly mention a perfect
example of an experimental system in which the regimes studied by us
could be easily realized, which is inspired by the beautiful Nobel
Lecture of Prof. G. Ertl \cite{ertl-nobel,ertl08}. Consider molecules
which undergo adsorption and desorption to and from a crystalline
surface. On the surface, the atoms may hop around and reversibly stick
to each other to form monolayer aggregates \cite{wintterlin97}. Then,
within the framework of phase field models, this process may be
described by the following evolution equation for the adsorbate
density fraction $\phi$ \cite{glotzer95}:
\begin{eqnarray}
  \label{eq:catal}
  \phi_t = M \Delta (W'(\phi) - g \Delta \phi) + k_+ (1 - \phi) - k_-
  \phi, 
\end{eqnarray}
where $W$ is a double-well potential with two minima between $\phi =
0$ and $\phi = 1$, $g$ is the short-range coupling constant, $M$ is a
kinetic coefficient, and $k_\pm$ are the adsorption and desorption
rates, respectively. Note that this equation can be rewritten as
\begin{eqnarray}
  \label{eq:32}
  \phi_t = M \Delta \{W'(\phi) - g \Delta \phi + k G_0 * (\phi
  - \bar \phi) \},
\end{eqnarray}
where $k = (k_+ + k_-)/M$, $\bar \phi = k_+ /(k_+ + k_-)$, and ``$*$''
denotes convolution in space, with $G_0$ given by (\ref{eq:G0}),
provided the spatial average of the initial data is $\bar \phi$. Upon
suitable rescaling, this is precisely the $H^{-1}$ gradient flow for
the energy $\mathcal E$, i.e., we have $u_t = \Delta (\delta \mathcal
E / \delta u)$, where $u$ is a rescaling of $\phi$. In particular,
minimizers of $\mathcal E$ are ground states of the considered system
in equilibrium in the mean-field limit. We note that the adsorption
and desorption rates $k_\pm$ can be quite small compared to the
hopping rate, resulting in very small values of $\ep \sim
k^{1/2}$. Therefore, one can achieve a very good scale separation
between the interfacial thickness (atomic scales) and the size of
adsorbate clusters (micro-scale) in this experimental setup.

Our paper is organized as follows. In Sec. \ref{sec:heur}, we present
heuristic arguments and give the statements of main results, in
Sec. \ref{sec:sharp}, we perform a detailed analysis of the sharp
interface energy $E$, in Sec. \ref{sec:conn} we establish a connection
between the sharp interface energy $E$ and the diffuse interface
energy $\mathcal E$. Finally, in Sec. \ref{sec:proofs} we conclude the
proofs of the theorems.

Throughout the paper, the symbols $L^p$, $H^k$, $W^{k,p}$,
$C^{k,\alpha}$, $BV$ denote the usual function spaces, $|\cdot|$
denotes the $d$-dimensional Lebesgue measure or the
$(d-1)$-dimensional Hausdorff measure of a set, depending on the
context, and $C$, $c$, etc., denote generic positive constants that
can change from line to line. The symbols $O(1)$ and $o(1)$ denote, as
usual, uniformly bounded and uniformly small quantities, respectively,
in the limit $\ep \to 0$, etc. Finally, we will say that a statement
holds for $\ep \ll 1$, etc., if there exists $\ep_0 > 0$ such that
that statement is true for all $0 < \ep \leq \ep_0$. For simplicity of
notation, the subscript $\ep$ is omitted for all quantities depending
on $\ep$.

\section{Heuristics and main results}
\label{sec:heur}

Let us begin our investigation by setting $d = 2$ and making a
simplifying assumption that the domain $\Omega$ is a torus: $\Omega =
[0,1)^2$. Let us also specify the domains of definition for the
functionals $\mathcal E$ and $E$. Formally, the diffuse interface
energy $\mathcal E[u]$ will be defined for all $u \in H^1(\Omega)$
subject to $\int_\Omega u\, dx = \bar u$, whereas the sharp interface
energy $E[u]$ will be defined for all $u \in BV(\Omega; \{-1, 1\})$.

The assumption that $\Omega$ is a torus, which is common in the
considered class of problems, eliminates the need to deal with the
boundary effects and, even more importantly, restores the
translational invariance inherent in the problem on the whole of
$\mathbb R^d$ (note that the choice of the size of $\Omega$ is
inconsequential, the obtained energy of the minimizers scales linearly
with $|\Omega|$). As a result, the kernel of the non-local part of the
energy becomes a function of $x - y$ only. With a slight abuse of
notation, in the following we will, therefore, replace $G(x, y)$ with
$G(x - y)$ everywhere below.

On heuristic grounds one would expect that the minimizers of $E$ at
$\ep \ll 1$ would be periodic with period $R \sim \ep^{1/3}$, whenever
$|\bar u| < 1$ and $|\bar u|$ is not too close to 1
\cite{ohta86,m:phd,m:pre02,choksi01}. A simple scaling analysis shows
that in this case $E \sim \ep^{2/3}$ as $\ep \to 0$ with $\bar u$
fixed. Our first result gives a justification for this energy scaling
without any assumptions about the minimizers (for statements about
existence and regularity of minimizers, see the following sections).

\begin{theorem}
  \label{t:cC}
  Let $W$ satisfy the assumptions (i)--(iv) at the beginning of
  Sec. \ref{sec:conn}, and let $\bar u \in (-1, 1)$ be fixed. Then
  there exist $\ep_0 > 0$ and $C > c > 0$, such that
  \begin{eqnarray}
    \label{eq:cC}
    c \ep^{2/3} \leq \min E, \, \min \mathcal E \leq C \ep^{2/3}
  \end{eqnarray}
  for all $\ep \leq \ep_0$.
\end{theorem}
Observe also that for $E$ this result still holds when $\Omega = [0,
1)^d$ for any $d$, while for $\mathcal E$ it holds at least for $d <
6$ (see Sec. \ref{sec:conn}). We note that for $\bar u = 0$ and $|u|
\leq 1$ such a result was obtained by Choksi, using somewhat different
techniques \cite{choksi01}. On the level of $E$ (with $\kappa = 0$),
Alberti, Choksi and Otto recently proved, among many other interesting
results, a stronger statement that in the limit $\ep \to 0$, the
constants in the upper and lower bounds in (\ref{eq:cC}) can be chosen
to be arbitrarily close to each other \cite{alberti09}. We note that
the case $\kappa = 0$ and $\bar u \in (-1,1)$ fixed can be treated as
the limit of energy $E$ considered by us as $\kappa \to 0$, when the
constraint $\int_\Omega u \, dx = \bar u$ gets automatically enforced
(see (\ref{eq:44})). 

Thus, when $\bar u \in (-1,1)$ is fixed, the energy $E$ admits a
non-trivial minimizer, whose energy scales as in (\ref{eq:cC}) when
$\ep \ll 1$. What about the case $|\bar u| > 1$? Here, in fact, it is
easy to see that the only minimizers admitted by $E$ are the trivial
ones. Consider, for example, the case $\bar u < -1$, the other case is
equivalent by symmetry. In this case the problem admits the unique
global minimizer $u = -1$. To see this, let us introduce the
characteristic function $\chi_{\Omega^+}$ of the set $\Omega^+ = \{ u
= +1 \}$ for a given $u \in BV(\Omega; \{-1, 1\})$. Then $u = 2
\chi_{\Omega^+} - 1$, and by a straightforward computation
\begin{eqnarray}
  \label{eq:5}
  E[u] & \geq & \frac12 \int_\Omega \int_\Omega (2 \chi_{\Omega^+}(x)
  - 1 - \bar u) G(x - y) (2 \chi_{\Omega^+}(y) - 1
  - \bar u) dx d y \nonumber \\ 
  & \geq & {(1 + \bar u)^2 \over 2 \kappa^2} - {2 (1 + \bar u)
    \over \kappa^2} |\Omega^+|.
\end{eqnarray}
Thus, when $\bar u < -1$, the second term in the last inequality in
(\ref{eq:5}) is positive, hence, is minimized by $|\Omega^+| = 0$. But
in this case $u = -1$ attains equality in (\ref{eq:5}), so $u = -1$ is
the minimizer. Thus, when $|\bar u| > 1$, non-trivial minimizers of
$E$ do not exist, and, therefore, at $|\bar u| = 1$ a bifurcation
occurs in the limit $\ep \to 0$.

The main purpose of this paper is to investigate the transition
between the trivial and the non-trivial minimizers of $E$ and
$\mathcal E$ that occurs in the neighborhood of $|\bar u| = 1$ for
$\ep \ll 1$. The energy $E$ captures most of the difficulty associated
with the considered problem. Therefore, we will spend most of our
effort in this paper to the studies of $E$ (see
Sec. \ref{sec:sharp}). At the same time, as we show later (see
Sec. \ref{sec:conn}), the statements about the behavior of $\min E$
also extend to that of $\min \mathcal E$ for $\ep \ll 1$ (the
correspondence of minimizers of the two energies will be a subject of
future study).

When $\Omega = [0, 1)^2$, the kernel $G$ has an explicit
representation
\begin{eqnarray}
  \label{eq:G}
  G(x) = \frac{1}{2 \pi} \sum_{\mathbf n \in \mathbb Z^2}
  K_0(\kappa |x  - \mathbf n|),
\end{eqnarray}
where $K_0$ is the modified Bessel function of the first kind. In
particular, $G > 0$ and we have the following asymptotic expansion
from the power series representation of $K_0$ \cite{abramowitz}:
\begin{eqnarray}
  \label{eq:Gest}
  G(x) = -\frac{1}{2 \pi} \ln (\bar \kappa |x|) + 
  O(|x|),  \qquad |x|  \ll   1, 
\end{eqnarray}
where 
\begin{eqnarray}
  \label{eq:kbar}
  \bar \kappa = \tfrac12 \kappa \exp \left( \gamma - \sum_{\mathbf
      n \in \mathbb Z^2 \backslash \{0\}} K_0(\kappa |\mathbf n|)
  \right),  
\end{eqnarray}
and $\gamma \approx 0.5772$ is the Euler's constant. We also have
$G(x)$ bounded whenever $|x| > \delta$, for any $\delta > 0$, and
(\ref{eq:Gest}) can be used to estimate derivatives of $G$ to $O(|x|
\ln |x|)$ as well.

Consider the case in which the value of $\bar u$ approaches $\bar u =
-1$ from above, with $\ep \ll 1$ fixed. Clearly, for large enough
deviations there exists a non-trivial minimizer. As can be seen from
the arguments in the proof of Theorem \ref{t:cC}, the size of the set
where $u = 1$ on the minimizer goes to zero as $\bar u \to
-1$. Heuristically, one would, therefore, expect that in this
situation the minimizer will consist of a number of isolated droplets
where $u = +1$ of small size in the background where $u =
-1$. Moreover, since on the scale of a droplet the interfacial energy
will give a dominant contribution, these droplets are expected to be
nearly circular. This motivates an introduction of the following
reduced energy:
\begin{eqnarray}
  \label{eq:EN}
  E_N(\{r_i\}, \{x_i\}) & = & \sum_{i = 1}^N \Bigl( 2 \pi \ep r_i 
    - 2 \pi (1 + \bar u) \kappa^{-2}
    r_i^2 - \pi  r_i^4 (\ln \bar \kappa r_i - \tfrac14) \Bigr)
  \nonumber \\  
  & + & 4 \pi^2 \sum_{i = 1}^{N-1} \sum_{j = i + 1}^N
  G( x_i - x_j) r_i^2 r_j^2.
\end{eqnarray}
which describes the energy of interaction of $N$ well separated
disk-shaped droplets of radius $r_i$ centered at $x_i$, to the leading
order. More precisely, the first term \eqref{eq:EN} stands for the
interfacial energy of all the droplets, the second term is the energy
of interactions between the droplets and the background, the third
term is the self-interaction energy of each droplet, and the last term
is the interaction energy of each droplet pair (for the case of a
single droplet in $\mathbb R^2$, see \cite{m:pre02}).

We can use the reduced energy in \eqref{eq:EN} to obtain the leading
order scaling of various quantities for $\ep \ll 1$ by balancing
different terms. From the balance of interfacial energy and the
self-interaction energy, one should have $r_i = O( \ep^{1/3} |\ln
\ep|^{-1/3})$. Balancing this with the second term leads, in turn, to
\begin{eqnarray}
  \label{eq:dbar}
  \bar \delta = \ep^{-2/3} |\ln \ep|^{-1/3} (1 + \bar u)
\end{eqnarray}
being an $O(1)$ quantity. Similarly, balancing the last term with the
first three leads to $N = O(|\ln \ep|)$, and the expected behavior of
$\min E_N = O(\ep^{4/3} |\ln \ep|^{2/3})$. One would also expect that,
since the droplets repel each other, in a minimum energy configuration
they would become uniformly distributed throughout $\Omega$.

Our main result proves and further quantifies this heuristic picture
on the level of the sharp interface energy $E$.

\begin{theorem}
  \label{t:main}
  Let $\bar u = -1 + \ep^{2/3} |\ln \ep|^{1/3} \bar\delta$, with some
  $\bar\delta > 0$ fixed. Then for any $\sigma > 0$ sufficiently small
  there exists $\ep_0 > 0$ such that for all $\ep \leq \ep_0$:
  \begin{itemize}
  \item[(i)] If $\bar \delta < \tfrac12 \sqrt[3]{9} \, \kappa^2$, then
    $u = -1$ is the unique global minimizer of $E$, with $\ep^{-4/3}
    |\ln \ep|^{-2/3} \min E = \tfrac12 \kappa^{-2} \bar\delta^2$.

  \item[(ii)] If $\bar\delta > \tfrac12 \sqrt[3]{9} \, \kappa^2$,
    there exists a non-trivial minimizer of $E$. The minimizer is
    \begin{eqnarray}
      \label{eq:balls}
      u(x) = - 1 + 2 \sum_{i = 1}^N \chi_{\Omega^+_i}(x),
    \end{eqnarray}
    where $\chi_{\Omega^+_i}$ are characteristic functions of $N$
    disjoint simply connected sets $\Omega^+_i \subset \Omega$ with
    boundary of class $C^3$, and $N = O(|\ln \ep|)$. The boundary of
    each set $\Omega^+_i$ is $O(\ep^{2/3 - \sigma})$-close (in
    Hausdorff sense) to a circle of radius $r_i$ centered at $x_i$.
    Furthermore,
    \begin{eqnarray}
      \min E = \tfrac12 \ep^{4/3} |\ln \ep|^{2/3} \kappa^{-2} \bar
      \delta^2 + E_N(\{ r_i\}, \{ x_i \}) + O(\ep^{5/3 - \sigma}), 
    \end{eqnarray}
    with $E_N = O( \ep^{4/3} |\ln \ep|^{2/3})$, $r_i = O(\ep^{1/3}
    |\ln \ep|^{-1/3})$,
    and
    \begin{eqnarray}
      \label{xijal}
      |x_i - x_j| > \ep^\sigma, \quad \forall j \not= i.
    \end{eqnarray}
  \item[(iii)] If $\bar\delta > \tfrac12 \sqrt[3]{9} \, \kappa^2$, in
    the limit $\ep \to 0$ we have
   \begin{eqnarray}
     \label{eq:lims}
     \ep^{-1/3} |\ln \ep|^{1/3} r_i & \to & \sqrt[3]{3}
   \end{eqnarray}
   uniformly,
   \begin{eqnarray}
     {1 \over |\ln \ep|} \sum_{i = 1}^N \delta(x - x_i) & \to & {1
       \over 2 \pi \sqrt[3]{9}} \left(
     \bar\delta - {\sqrt[3]{9} \over 2} \, \kappa^2 \right),
  \end{eqnarray}
  weakly in the sense of measures, and
   \begin{eqnarray}
     \ep^{-4/3} |\ln \ep|^{-2/3} \min E & \to & {\sqrt[3]{9} \over 2}
     \left( \bar \delta - {\sqrt[3]{9} \over 4} \, \kappa^2 \right).
   \end{eqnarray} 
\end{itemize}
\end{theorem}
Note that a more detailed result on the structure of the transition
occurring near $\bar \delta = \tfrac12 \sqrt[3]{9} \, \kappa^2$ is
presented in Proposition \ref{p:E33}.

Let us make a few remarks related to the statements of Theorem
\ref{t:main}. For small, but finite values of $\ep$ this theorem
establishes an equivalence between the sharp interface energy $E$ and
the energy of $N$ interacting droplets $E_N$, in the sense that the
minimizers of $E$ are close to ``almost'' minimizers of $E_N$, i.e.,
we have $E_N < \min E_N + O(\ep^{5/3 - \sigma})$. Nevertheless, to
prove closeness of minimizers of $E$ to those of $E_N$ we also need
some coercivity of the energy $E_N$. This problem has to do with the
properties of the minimizers of the pairwise interaction of the
droplets, i.e. the choice of $x_i$ which minimize $E_N$ with fixed
$r_i$. This becomes a difficult problem in the case of interest, since
we generally expect $N \gg 1$ (for a numerical solution at a few
values of $N$ and $\kappa = 2$ see Fig. \ref{fig:minima}).  It would
be natural to conjecture that at small enough $\ep$ the minimizing
droplets will arrange themselves into a periodic lattice close to a
hexagonal (close-packed) lattice. Proving this kind of result,
however, is a major challenge (see \cite{theil06} for a recent proof
for a certain class of pair interactions), which is one of the open
questions also in many other problems, such as the problem of
characterizing Abrikosov vortex lattices, for example
\cite{aftalion07}. Let us mention here a recent result by Chen and
Oshita, who proved that in the case $\kappa = 0$ the hexagonal
arrangement of disks is energetically the best among simple periodic
lattices \cite{chen07arma}. Yet, it is not known if the same result
also holds for more general arrangements of droplets. Here we prove a
weaker result that the number density of droplets becomes
asymptotically uniform as $\ep \to 0$, leading also to uniform
distribution of energy (compare with \cite{alberti09}). Moreover, we
identify the precise asymptotic behavior of the minimal energy and the
size of the minimizing droplets as $\ep \to 0$.

\begin{figure}
  \centering
  \includegraphics[width=12cm]{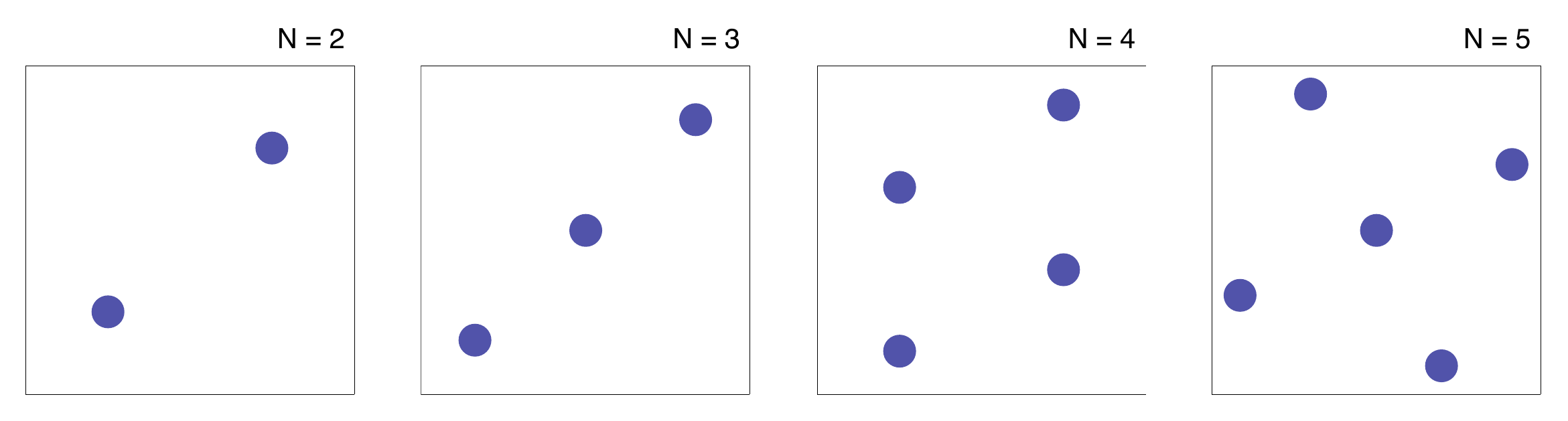}  
  \caption{``Coulombic dice'': Minimizers of $E_N$ with $r_i =
    \sqrt[3]{3} \, \ep^{1/3} |\ln \ep|^{-1/3}$ for $\kappa = 2$ and $N
    = 2, 3, 4, 5$, obtained using random search algorithm.}
  \label{fig:minima}
\end{figure}

Lastly, we establish the asymptotic behavior of the minimal value of
the diffuse interface energy $\mathcal E$ (which, of course, agrees
with the result for the sharp interface energy).

\begin{theorem}
  \label{t:ok}
  Let $W$ satisfy assumptions (i)--(iv) at the beginning of
  Sec. \ref{sec:conn}, let $\bar u = -1 + \ep^{2/3} |\ln \ep|^{1/3}
  \bar\delta$, with some $\bar\delta > 0$ fixed, and let $\kappa$ be
  given by (\ref{eq:kappa}). Then
  \begin{itemize}
  \item[(i)] If $\bar \delta \leq \tfrac12 \sqrt[3]{9} \, \kappa^2$,
    then $\ep^{-4/3} |\ln \ep|^{-2/3} \min \mathcal E \to \tfrac12
    \kappa^{-2} \bar\delta^2$,

  \item[(ii)] If $\bar\delta > \tfrac12 \sqrt[3]{9} \, \kappa^2$, then
    $ \ep^{-4/3} |\ln \ep|^{-2/3} \min \mathcal E \to {\sqrt[3]{9}
      \over 2} \left( \bar \delta - {\sqrt[3]{9} \over 4} \, \kappa^2 
    \right)$, 
  \end{itemize}
  as $\ep \to 0$.
\end{theorem}

The proofs of Theorems \ref{t:cC}--\ref{t:ok} are based on a number of
propositions established in Secs. \ref{sec:sharp} and \ref{sec:conn},
and are completed in Sec. \ref{sec:proofs}.

\section{Analysis of  the sharp interface problem}
\label{sec:sharp}

Our plan for the analysis of the sharp interface problem consists of a
number of steps which we list below:

\begin{enumerate}
\item Introduce a suitably rescaled energy $\bar E$ and domain $\bar
  \Omega$.

\item Establish existence and regularity of the minimizers of $\bar E$
  (subsets of $\bar \Omega$ where $u = 1$).

\item Establish some a priori estimates for the geometry of the
  minimizers of $\bar E$ and uniform bounds on the induced long-range
  potential.

\item Establish that different connected components of minimizers of
  $\bar E$ are separated by large distances in $\bar \Omega$.

\item Establish that each connected component of a minimizer of $\bar
  E$ is close to a disk (hence the term ``droplet'').

\item Establish equivalence between $\min \bar E$ and $\min \bar E_N$
  (the suitably rescaled version of $E_N$).

\item Improve the estimate for the separation distance between
  different droplets.

\item Prove uniform convergence of the rescaled droplet radii to a
  universal constant.

\item Prove convergence of $\min \bar E$ to a limit and convergence of
  the normalized droplet density in the original, unscaled domain
  $\Omega$ to a limit, as $\ep \to 0$.
\end{enumerate}
This plan is carried out in the rest of this section via a series of
lemmas and propositions.

\subsection{Scaling}
\label{sec:scaling}

We begin by introducing a suitable rescaling, in which the main
quantities of interest become $O(1)$ quantities in the limit $\ep \to
0$. Motivated by the discussion of Sec. \ref{sec:heur}, we define the
rescaled energy $\bar E$ (with the energy of the uniform state $u =
-1$ subtracted) and a new coordinate $\bar x \in \bar \Omega = [0,
\ep^{1/3} |\ln \ep|^{-1/3})^2$, where $\bar \Omega$ is a
two-dimensional torus with period $\ep^{1/3} |\ln \ep|^{-1/3}$:
\begin{eqnarray}
  \label{eq:10}
  E = \ep^{4/3} |\ln \ep|^{2/3}
  \, (\tfrac12 \kappa^{-2} \bar \delta^2 + \bar E), \qquad x = 
  {\ep^{1/3} \over |\ln \ep|^{1/3}} \, \bar x.
\end{eqnarray}
The energy $\bar E$ can be conveniently expressed in term of the set
$\bar \Omega^+ \subset \bar \Omega$ in which $u = 1$:
\begin{eqnarray}
  \label{eq:Ebar}
  \bar E & = &|\ln \ep|^{-1} \bigl(  |\partial \bar \Omega^+| - 2
  \bar \delta \kappa^{-2} |\bar \Omega^+| \bigr)  \nonumber \\  
  & + & 2 |\ln \ep|^{-2} \int_{\bar \Omega^+} \int_{\bar \Omega^+}  G
  \bigl( \ep^{1/3} |\ln \ep|^{-1/3} (\bar x - \bar y) \bigr) d \bar x
  \, d\bar y.   
\end{eqnarray}

We also need an expression for the rescaled energy $\bar E_N$ of a
system of interacting droplets. With the help of (\ref{eq:10}), we can
write the rescaling of (\ref{eq:EN}) as
\begin{eqnarray}
  \label{eq:Enb}
  \bar E_N & = &  {2 \pi \over |\ln \ep| } \sum_{i = 1}^N \left\{ \bar
    r_i - \bar \delta \kappa^{-2} \, 
    \bar r_i^2 - \tfrac12 |\ln \ep|^{-1} \bar r_i^4 \left(\ln
      (\ep^{1/3} |\ln  \ep|^{-1/3} 
      \bar \kappa \bar r_i) - \tfrac14 \right) \right\} \nonumber \\ 
  & + & {4 \pi^2 \over |\ln \ep|^2 } \sum_{i = 1}^{N-1} \sum_{j =
    i+1}^N G( \ep^{1/3} |\ln \ep|^{-1/3}  (\bar x_i - \bar x_j))
  \bar r_i^2 \bar r_j^2,
\end{eqnarray}
where $\bar r_i$ and $\bar x_i$ are the radii and the centers of the
droplets, respectively.

\subsection{Properties of minimizers}
\label{sec:lb}

Let us begin with the statement of a result on the existence and
regularity of minimizers of $\bar E$ (or, equivalently, of $E$), which
is obtained by straightforwardly adapting the results of
\cite{massari74} for sets of prescribed mean curvature.

\begin{proposition}
  \label{p:exist}
  There exists a set $\bar \Omega^+$ of finite perimeter which
  minimizes $\bar E$ in (\ref{eq:Ebar}). The boundary $\partial
  \bar\Omega^+$ of this set is a curve of class $C^{1,\alpha}$ for
  some $\alpha \in (0, 1)$.
\end{proposition}

\noindent In view of this, in the following we will always assume that
minimizers $\bar\Omega^+$ of $\bar E$ are closed sets. We also note
that 
\begin{eqnarray}
  \label{eq:v}
  v(\bar x) = |\ln \ep|^{-1} \int_{\bar \Omega^+} G(\ep^{1/3} |\ln
  \ep|^{-1/3} (\bar x - \bar y)) \, d \bar y
\end{eqnarray}
is in $W^{2,p}(\bar \Omega)$, with any $p > 1$, and, hence, in
$C^{1,\alpha}(\bar \Omega)$ for any $\alpha \in (0, 1)$. Indeed, $\bar
v$ solves the equation
\begin{eqnarray}
  \label{eq:vbeq}
  -\Delta v + \ep^{2/3} |\ln \ep|^{-2/3} \kappa^2 v = |\ln \ep|^{-1}
  \chi_{\bar\Omega^+}, 
\end{eqnarray}
where $\chi_{\bar\Omega^+}$ is the characteristic function of
$\bar\Omega^+$, in $\bar \Omega$, and so the result follows by
standard elliptic regularity theory \cite{gilbarg}.  As a consequence,
we have a higher regularity for the boundary of the minimizer $\bar
\Omega^+$ of $\bar E$ (\cite{giusti}, see also \cite{roger08}):
\begin{corollary}
  \label{c:c3}
  The boundary $\partial \bar\Omega^+$ of a minimizer $\bar \Omega^+$
  of $\bar E$ is of class $C^{3,\alpha}$.
\end{corollary}

\noindent Note that this regularity result also holds more generally
for local minimizers of $\bar E$ in dimensions $d \leq 7$
\cite{massari74}, hence, in particular, the expressions for the first
and second variation of $\bar E$ in $d \leq 3$ obtained in
\cite{m:pre02} are justified (see also \cite{choksi07} for the case of
arbitrary dimensions). If $\rho \in C^1(\partial \bar\Omega)$, $a >
0$, and $\bar\Omega_a$ is the set obtained by displacing $\partial
\bar\Omega^+$ by $a \rho$ in the outward normal direction, then $a
\mapsto \bar E(\bar \Omega^+_a)$ is twice continuously differentiable
at $a = 0$, and we have \cite{m:pre02} (for the reader's convenience,
the computation is reproduced in Appendix \ref{a:vars}):
\begin{eqnarray}
  \label{eq:Ebvar1}
  |\ln \ep| \left. {d \bar E (\bar \Omega^+_a) \over da} \right|_{a =
    0} = \int_{\partial 
    \bar\Omega^+} ( K(\bar x) -2 \bar\delta \kappa^{-2} + 4 v(\bar x))
  \rho(\bar x)  \, d \mathcal H^1(\bar x),   \\ 
  |\ln \ep| \left. {d^2 \bar E  (\bar \Omega^+_a) \over da^2}
  \right|_{a = 0} = \int_{\partial \bar\Omega^+}  \bigl( |\nabla
  \rho(\bar x) |^2 + 4 \nu(\bar x) \cdot
  \nabla v(\bar x) \, \rho^2(\bar x) \bigr) \, d \mathcal H^1(\bar x) 
  \nonumber \\  
  + \int_{\partial \bar\Omega^+} (4 v (\bar x) - 2 \bar
  \delta \kappa^{-2}) K(\bar x)  \rho^2(\bar x)  \, d
  \mathcal  H^1(\bar  x) \nonumber \\   
  + 4 |\ln \ep|^{-1} \int_{\partial \bar\Omega^+} \int_{\partial
    \bar\Omega^+} 
  G(\ep^{1/3} |\ln  \ep|^{-1/3} (\bar x - \bar y)) \rho(\bar x)
  \rho(\bar y) \, d \mathcal
  H^1(\bar x) d  \mathcal H^1(\bar y).
  \label{eq:Ebvar2}
\end{eqnarray}
where $K(\bar x)$ is the curvature at point $\bar x \in \partial
\bar\Omega^+$, with the sign convention that $K > 0$ if $\bar\Omega^+$
is convex, and $\nu(\bar x)$ is the outward unit normal to $\partial
\bar\Omega^+$ at that point. The associated Euler-Lagrange equation
for $\partial \bar\Omega^+$ reads
\begin{eqnarray}
  \label{eq:ELsh}
  K(\bar x) = 2 \bar\delta \kappa^{-2} - 4 v(\bar x),
\end{eqnarray}
which also allows to simplify the expression in (\ref{eq:Ebvar2})
evaluated on a minimizer to
\begin{eqnarray}
  \label{eq:Ebvar3}
  |\ln \ep| \left. {d^2 \bar E  (\bar \Omega^+_a) \over da^2}
  \right|_{a = 0}  = \hspace{6.5cm} \nonumber \\ 
  \int_{\partial \bar\Omega^+}  \bigl( |\nabla
  \rho(\bar x) |^2 + 4 \nu(\bar x) \cdot
  \nabla v(\bar x) \, \rho^2(\bar x) - K^2(\bar x) \rho^2(\bar x) \bigr)
  \, d  \mathcal H^1(\bar x) \nonumber \\ 
  + 4  |\ln \ep|^{-1} \int_{\partial \bar\Omega^+} \int_{\partial
    \bar\Omega^+} 
  G(\ep^{1/3} |\ln  \ep|^{-1/3} (\bar x - \bar y)) \rho(\bar x)
  \rho(\bar y) \, d \mathcal
  H^1(\bar x) d  \mathcal H^1(\bar y). 
\end{eqnarray}

We will use these equations later on to establish some properties of
the minimizers for $\ep \ll 1$.  Meanwhile, let us begin our analysis
with some basic estimates.

\begin{lemma}
  \label{l:bbounds}
  Let $\bar\Omega^+$ be a minimizer of $\bar E$. Then there exists $C
  > 0$ such that 
  \begin{eqnarray}
    |\bar \Omega^+| & \leq & C |\ln \ep|, \label{omp0} \\
    |\partial \bar \Omega^+| & \leq & C |\ln \ep|. \label{domp0} 
  \end{eqnarray}
  for $\ep \ll 1$.
\end{lemma}
\begin{proof}
  First of all, by representation (\ref{eq:G}) we have $G(x - y) \geq
  c > 0$ for all $x, y \in \Omega$. Therefore, in view of the fact
  that $\min \bar E \leq 0$ (since $\bar E = 0$ if $\bar\Omega^+ =
  \varnothing$), from (\ref{eq:Ebar}) we have
  \begin{eqnarray}
    \label{eq:14}
    0 \geq |\ln \ep| \, \bar E & \geq & 2 |\ln \ep|^{-1} \int_{\bar
      \Omega^+} \int_{\bar 
      \Omega^+} G\bigl( \ep^{1/3} |\ln \ep|^{-1/3} (\bar x - \bar y)
    \bigr) d \bar x  - 2  \bar \delta \kappa^{-2} |\bar \Omega^+|
    \nonumber \\  
    & \geq & 2 c |\ln \ep|^{-1} \, |\bar \Omega^+|^2 - 2  \bar \delta 
    \kappa^{-2} |\bar \Omega^+|,  
  \end{eqnarray}
  which gives (\ref{omp0}).  On the other hand, we also have
  \begin{eqnarray}
    \label{eq:12}
    |\partial \bar\Omega^+| \leq 2 
    \bar \delta \kappa^{-2} |\bar \Omega^+|.
  \end{eqnarray}
  Therefore, from (\ref{omp0}) we immediately obtain (\ref{domp0}).
\end{proof}

As a corollary, it follows from (\ref{domp0}) that the diameter of
each connected subset $\bar \Omega^+_i$ of $\bar \Omega^+$ is bounded
by $O(|\ln \ep|)$
\begin{eqnarray}
  \label{eq:diamln}
  \mathrm{diam} (\bar \Omega^+_i) \leq C |\ln \ep|,
\end{eqnarray}
for some $C > 0$ independent of $\ep \ll 1$.

Our next step is to show that the area of each connected component of
$\bar\Omega^+ \not = \varnothing$ is uniformly bounded above and below
independently of $\ep$.

\begin{lemma}
  \label{l:EcC}
  Let $\bar\Omega^+ = \cup_{i=1}^N \bar\Omega^+_i$ be a non-trivial
  minimizer of $\bar E$, where $\bar\Omega^+_i$ are the disjoint
  connected components of $\bar\Omega^+$. Then, there exist  $C > c >
  0$ such that
  \begin{eqnarray}
    \label{eq:OmcC}
    c \leq |\bar\Omega^+_i|,  |\partial \bar\Omega^+_i| \leq C, \qquad
    \mathrm{diam} (\bar\Omega^+_i) \leq C,
  \end{eqnarray}
  for $\ep \ll 1$.
\end{lemma}

\begin{proof}
  First, note that since by Corollary \ref{c:c3} the set $\partial
  \bar\Omega^+$ is of class $C^{3,\alpha}$ we have $N < \infty$. To
  see that (\ref{eq:OmcC}) holds, we first write $\bar E$ as
  \begin{eqnarray}
    \label{eq:Eb}
    |\ln \ep| \, \bar E & = & \sum_{i=1}^N \Biggl(|\partial
    \bar\Omega^+_i| 
    - 2 \bar \delta \kappa^{-2} |\bar\Omega^+_i| \nonumber \\ 
    & + & 2 |\ln \ep|^{-1} \int_{\bar\Omega^+_i} \int_{\bar\Omega^+_i}
    G(\ep^{1/3} |\ln \ep|^{-1/3}  (\bar x - \bar
    y)) \, d \bar x d \bar y \nonumber \\ 
    & + & 2 |\ln \ep|^{-1} \sum_{j \not= i} \int_{\bar\Omega^+_i}
    \int_{\bar\Omega^+_j}
    G(\ep^{1/3} |\ln \ep|^{-1/3}  (\bar x - \bar
    y)) \, d \bar x d \bar y\Biggr).
  \end{eqnarray}
  In view of (\ref{eq:diamln}) and (\ref{eq:Gest}), the integral in
  the second line in (\ref{eq:Eb}) is bounded from below by $
  \tfrac{1}{6 \pi} (1 - \delta) |\ln \ep| \, |\bar\Omega^+_i|^2$ for
  any $\delta > 0$, provided $\ep$ is small enough. Therefore,
  removing the set $\bar\Omega^+_i$ from $\bar\Omega^+$ will result in
  the change of energy $\Delta \bar E$ estimated as
  \begin{eqnarray}
    \label{eq:delE}
    |\ln \ep| \, \Delta \bar E & \leq & - \biggl(
    |\partial\bar\Omega^+_i|  - 2 \bar \delta \kappa^{-2}
    |\bar\Omega^+_i| +  \tfrac{1}{3 \pi} (1 - \delta)
    |\bar\Omega^+_i|^2  \biggr) \nonumber \\ 
    & \leq & -  \biggl(
    2 \sqrt{\pi} \, |\bar\Omega^+_i|^{1/2}  - 2 \bar \delta
    \kappa^{-2}  |\bar\Omega^+_i| +  \tfrac{1}{3 \pi} (1 - \delta) 
    |\bar\Omega^+_i|^2  \biggr),
  \end{eqnarray}
  where in the first line we took into account that $G > 0$ and in the
  second line used the isoperimetric inequality. Then, by direct
  inspection (see also Fig. \ref{fig:fa}) we have $\Delta \bar E < 0$,
  contradicting minimality of $\bar E$ on $\bar \Omega^+$, unless $c
  \leq |\bar\Omega^+_i| \leq C$ for some $C > c > 0$, independently of
  $\ep \ll 1$.  Finally, the lower bound for $ |\partial
  \bar\Omega^+_i|$ follows from the isoperimetric inequality, and the
  upper bound is obtained by applying the previous argument to the
  first line in (\ref{eq:delE}).
\end{proof}

\begin{figure}
  \centering
  \includegraphics[width=8cm]{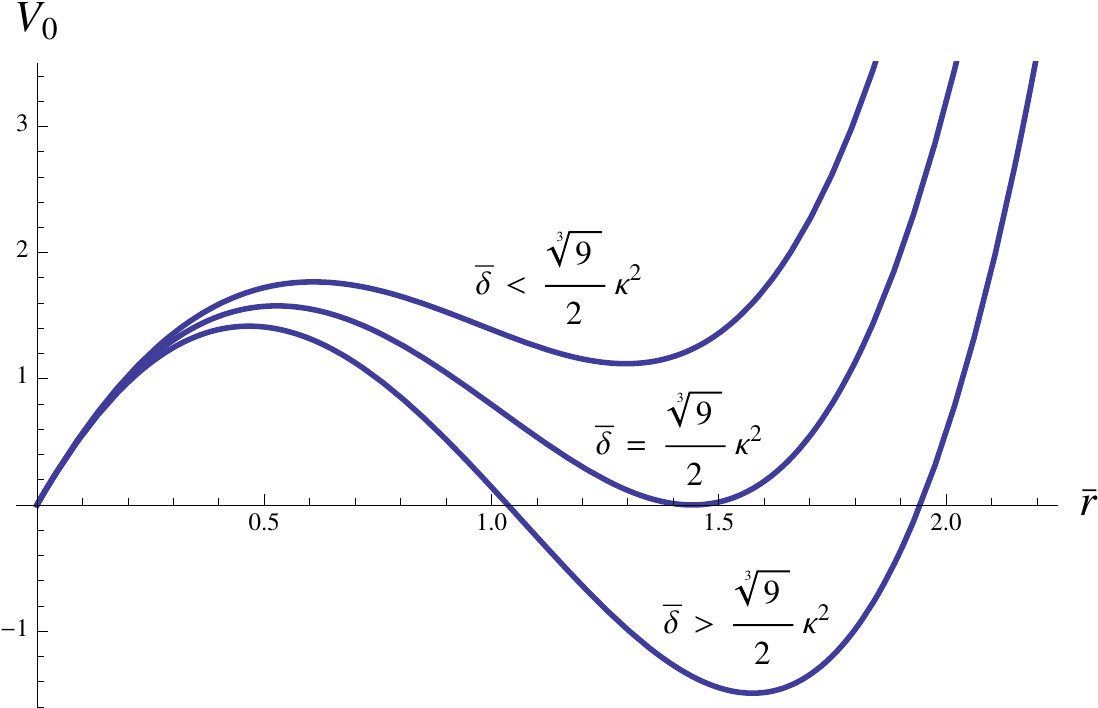}
  \caption{The graph of $V_0(\bar r)$ from (\ref{eq:Vbri}) for
    different values of $\bar \delta$.}
  \label{fig:fa}
\end{figure}

Following the same arguments, we also immediately arrive at the
following non-existence result:
\begin{proposition}
  \label{p:db} 
  Let $\bar \delta < \tfrac12 \sqrt[3]{9} \, \kappa^2$ be fixed. Then
  the unique minimizer of $\bar E$ is $\bar \Omega^+ = \varnothing$
  for $\ep \ll 1$.
\end{proposition}

\begin{proof}
  Let us introduce the function $V_v: [0, \infty) \to \mathbb R$,
  defined as
  \begin{eqnarray}
    \label{eq:Vbri}
    V_{v}(\bar r) = 2 \pi \left( \bar r +  (2 v - \bar \delta
      \kappa^{-2}) \bar r^2 + \tfrac16 \bar r^4 
    \right), 
  \end{eqnarray}
  whose graph at $v = 0$ and several values of $\bar\delta$ is shown
  in Fig. \ref{fig:fa}. If $\bar\Omega^+_i$ is a connected component
  of $\bar\Omega^+$ and $\bar r_i = (\tfrac{1}{\pi}
  |\bar\Omega^+_i|)^{1/2}$, then by the same arguments as in Lemma
  \ref{l:EcC}, the energy gained by removing $\bar\Omega^+_i$ from
  $\bar\Omega^+$ is bounded below by $|\ln \ep|^{-1} (V_0(\bar r_i) +
  o(1))$, as long as $\ep \ll 1$. Then, by direct inspection $V_0(\bar
  r)$ is always positive under the assumptions of the proposition,
  making $\bar\Omega^+ = \varnothing$ energetically preferred.
\end{proof}

\noindent Note that the asymptotic value of the threshold of $\bar
\delta$ in Proposition \ref{p:db} below which no non-trivial
minimizers are present was computed in \cite{m:pre02}.  Another simple
corollary to Proposition \ref{l:EcC} is the following
\begin{lemma}
  \label{l:N}
  Let $\bar\Omega^+$ be a non-trivial minimizer of $\bar E$ and let
  $N$ be the number of disjoint connected components of
  $\bar\Omega^+$. Then there exists $C > 0$ such that 
  \begin{eqnarray}
    \label{eq:1}
    N \leq C |\ln \ep|,
  \end{eqnarray}
  for $\ep \ll 1$.
\end{lemma}

Let us now establish a uniform bound on the potential $v$. Note that a
version of this result is also an important component in the proofs of
\cite{alberti09}.
\begin{lemma}
  \label{l:v}
  Let $\bar\Omega^+$ be a non-trivial minimizer of $\bar E$. Then for
  any $\alpha \in (0, 1)$ we have
  \begin{eqnarray}
    \label{eq:0vC}
    0 < v \leq C,    \qquad ||v||_{C^{1,\alpha}(\bar \Omega)} \leq
    C, 
  \end{eqnarray}
  where $v$ is given by (\ref{eq:v}), for some $C > 0$ independent of
  $\ep \ll 1$.
\end{lemma}

\begin{proof}
  We start by noting that $v > 0$ in view of positivity of $G$.
  Let us now estimate the gradient of $v$. Using (\ref{eq:Gest}) and
  Lemmas \ref{l:EcC} and \ref{l:N}, we get
  \begin{eqnarray}
    \label{eq:50}
    |\nabla v(\bar x)| \leq |\ln \ep|^{-1} \int_{\bar\Omega^+} |\nabla
    G(\ep^{1/3} |\ln \ep|^{-1/3} |\bar x - \bar y|)| \, d \bar y
    \nonumber \\ 
    \leq  |\ln \ep|^{-1} \int_{B_{\bar r}(\bar x)} |\nabla 
    G(\ep^{1/3} |\ln \ep|^{-1/3} |\bar x - \bar y|)| \, d \bar y
    \nonumber \\ + |\ln \ep|^{-1} \int_{\bar\Omega^+ \backslash B_{\bar
        r}(\bar x)}  |\nabla G(\ep^{1/3} |\ln \ep|^{-1/3} |\bar x -
    \bar y|)| \,  d \bar y \nonumber \\
    \leq C( |\ln \ep|^{-1} \bar r + \bar r^{-1}) \leq 2 C |\ln
    \ep|^{-1/2},
  \end{eqnarray}
  for some $C > 0$, where $B_{\bar r}(\bar x)$ is a disk of radius
  $\bar r$ centered at $\bar x$, and the last inequality is obtained
  by choosing $\bar r = |\ln \ep|^{1/2}$. Therefore, by the results of
  Lemma \ref{l:EcC}, we see that
  \begin{eqnarray}
    \label{eq:osc}
    \mathop{\mathrm{osc}}_{\bar x \in \bar\Omega^+_i} v(\bar x) = o(1),
  \end{eqnarray}
  for each connected component $\bar\Omega^+_i$ of $\bar\Omega^+$. To
  see that this implies the conclusion of the lemma, suppose that, to
  the contrary, we have $\max v = M \gg 1$. Since by (\ref{eq:vbeq})
  the function $v$ is subharmonic in $\bar\Omega \backslash
  \bar\Omega^+$, it achieves its maximum in the closure of some
  $\bar\Omega^+_i$. Therefore, in view of (\ref{eq:osc}) we have $v
  \geq \tfrac12 M$ in $\bar\Omega^+_i$. Then, following the same
  arguments as in the proof of Lemma \ref{l:EcC}, for large enough $M$
  we can lower the energy by removing $\bar\Omega^+_i$ from
  $\bar\Omega^+$.  

  Finally, by \cite[Theorem 9.11]{gilbarg} we have
  $||v||_{W^{2,p}(B_{1}(\bar x))} \leq C$, where $B_1(\bar x)$ is a
  disk of radius 1 centered at $\bar x \in \bar\Omega$, for some $C >
  0$ and any $p > 2$, independently of $\bar x$ and $\ep \ll
  1$. Hence, the uniform H\"older estimate on the gradient follows by
  Sobolev imbedding.
\end{proof}

We can also immediately conclude from (\ref{eq:ELsh}) and
(\ref{eq:osc}) that the curvature of $\partial\bar\Omega^+$ is
uniformly bounded both from above and below by positive constants,
implying that each $\bar\Omega^+_i$ is convex.  Note that this result
justifies the terminology ``droplet'' for each $\bar\Omega^+_i$ which
we will be using from now on.

\begin{lemma}
  \label{l:K}
  Let $\partial \bar\Omega^+$ be the boundary of a minimizer
  $\bar\Omega^+$ of $\bar E$. Then we have
  \begin{eqnarray}
    \label{eq:51}
    c \leq K(\bar x) \leq C,
  \end{eqnarray}
  for all $\bar x \in \partial \bar\Omega^+$, with some $C > c > 0$
  independent of $\ep \ll 1$. In particular, when $\ep \ll 1$, each
  connected component $\bar\Omega^+_i$ of $\bar\Omega^+$ is convex and
  simply connected.
\end{lemma}

\begin{proof}
  The upper bound is an immediate consequence of (\ref{eq:ELsh}) and
  positivity of $v$. To obtain the lower bound, let us note that by
  the results of Lemma \ref{l:EcC}, for every connected component
  $\bar\Omega^+_i$ there exists a disk $B_{\bar r_i}(\bar x_i)$, with
  $\bar r_i = O(1)$, such that $\bar\Omega^+_i \subset B_{\bar
    r_i}(\bar x_i)$. Therefore, translating $B_{\bar r_i}(\bar x_i)$
  until its boundary touches $ \partial \bar\Omega^+_i$, we obtain a
  point $\bar x_i' \in \partial \bar\Omega^+_i$, such that $K(\bar
  x_i') \geq \bar r_i^{-1} \geq 2 c$, for some $c > 0$ independent of
  $\ep \ll 1$. Now, by (\ref{eq:ELsh}) we have $v(\bar x_i') \leq
  \tfrac12 (\bar \delta \kappa^{-2} - c)$. At the same time, by
  (\ref{eq:osc}) this implies that $v(\bar x) \leq \tfrac14 (2 \bar
  \delta \kappa^{-2} - c)$ for all $\bar x \in \partial
  \bar\Omega^+_i$, which, again, by (\ref{eq:ELsh}) gives the
  statement.
\end{proof}

We now show that different connected components of $\bar\Omega^+$
cannot come too close to each other when $\ep \ll 1$.

\begin{lemma}
  \label{l:dist}
  Let $\bar\Omega^+ = \cup_{i=1}^N \bar\Omega^+_i$ be a non-trivial
  minimizer of $\bar E$, where $\bar\Omega^+_i$ are the disjoint
  connected components of $\bar\Omega^+$, and let $N \geq 2$. Then,
  there exists $C > 0$ such that
  \begin{eqnarray}
    \label{eq:52}
    \mathrm{dist} (\bar\Omega^+_i, \bar\Omega^+_j) \geq C \qquad
    \forall i \not = j,
  \end{eqnarray}
  for $\ep \ll 1$.
\end{lemma}

\begin{proof}
  Let $\bar x_i \in \bar\Omega^+_i$ and $\bar x_j \in \bar\Omega^+_j$
  be such that $r = |\bar x_i - \bar x_j| = \mathrm{dist}
  (\bar\Omega^+_i, \bar\Omega^+_j) > 0$. Consider a disk $B$ centered
  at $\tfrac12(\bar x_1 + \bar x_2)$ with radius $R = 2 r$ and a
  rectangle $Q$ inscribed into $B$ which is shown by the thick solid
  lines in Fig. \ref{fig:join}. In view of the uniform bound on the
  curvature of $\partial \bar\Omega^+$ obtained in Lemma \ref{l:K},
  the curve segments $\partial \bar\Omega^+_i \cap Q$ and $\partial
  \bar\Omega^+_j \cap Q$ passing through $\bar x_i$ and $\bar x_j$,
  respectively, intersect $\partial Q$ transversally as in
  Fig. \ref{fig:join} when $r \ll 1$. Furthermore, we have
  $\mathrm{dist}(\partial \bar\Omega^+_i \cap \partial Q^+, \partial
  \bar\Omega^+_j \cap \partial Q^+) \leq 2 r$ and
  $\mathrm{dist}(\partial \bar\Omega^+_i \cap \partial Q^-, \partial
  \bar\Omega^+_j \cap \partial Q^-) \leq 2 r$, where $\partial Q^+$
  and $\partial Q^-$ are the right and the left side of the boundary
  of the rectangle relative to the line through $\bar x_1$ and $\bar
  x_2$, respectively, for sufficiently small $r$ independent of $\ep
  \ll 1$ (see Fig. \ref{fig:join}). At the same time, we have
  $|\partial \bar\Omega^+_i \cap Q | + |\partial \bar\Omega^+_j \cap
  Q| \geq 4 r \sqrt{3}$. Therefore, reconnecting the points $\partial
  \bar\Omega^+_i \cap \partial Q^+$ with $\partial \bar\Omega^+_j
  \cap \partial Q^+$, and $\partial \bar\Omega^+_i \cap \partial Q^-$
  with $\partial \bar\Omega^+_j \cap \partial Q^-$ by straight lines
  and including the region between them into $\bar\Omega$, we will
  decrease $|\partial \bar \Omega^+|$ by at least $4 (\sqrt{3} - 1)
  r$. Thus, the change $\Delta \bar E$ in the total energy is
  estimated to be
  \begin{eqnarray}
    \label{eq:53}
    |\ln \ep| \Delta \bar E & \leq & -4  (\sqrt{3} - 1) r + 4 
    \int_Q v(\bar x) \, d\bar x + \nonumber \\  
    & + & 2 |\ln \ep|^{-1} \int_Q \int_Q G(\ep^{1/3} |\ln \ep|^{-1/3}
    (\bar x -  \bar y) ) \, d\bar x  d \bar y.
  \end{eqnarray}
  Finally, in view of Lemma \ref{l:v} and (\ref{eq:Gest}), the the
  right-hand side of (\ref{eq:53}) is bounded above by $-C_1 r + C_2
  r^2$, with $C_{1,2} > 0$ independent of $\ep \ll 1$. Hence, the
  energy of such a rearrangement will be lower if $r$ is sufficiently
  small, for all $\ep \ll 1$.
\end{proof}

\begin{figure}
  \centering
  \includegraphics[width=4cm,angle=-90]{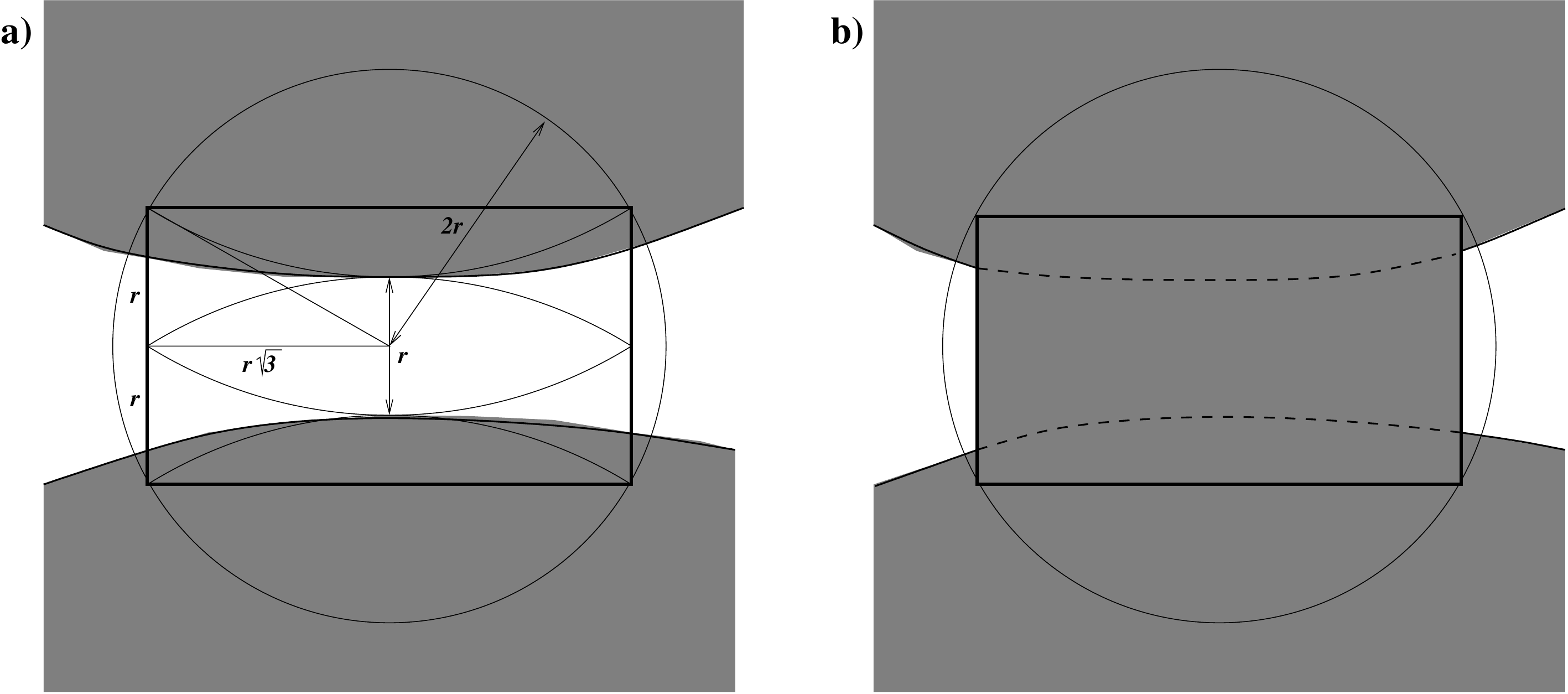}
  \caption{Schematics of the rearrangement argument of Lemma
    \ref{l:dist}. In (a), the set $\bar\Omega^+$ is shown in gray,
    solid arcs show the bounds on the location of $\partial
    \bar\Omega^+$, the thick solid lines show the rectangle $Q$. In
    (b), the gray region shows the rearranged $\bar\Omega^+$.}
  \label{fig:join}
\end{figure}

As our next step, we establish that different droplets must, in fact,
be sufficiently far from each other. We note that this result is a
manifestation of the ``repumping'' instability, which does not allow
two droplets to approach each other sufficiently closely. Dynamically,
this instability results in the growth of one droplet at the expense
of its neighbor shrinking. This instability mechanism for
reaction-diffusion systems was first pointed out in \cite{ko:mk85}
(see also \cite{ko:book}) and further studied in the context of
two-dimensional periodic structures in \cite{m1:prl97,m:phd,m:pre02}.

\begin{lemma}
  \label{l:al}
  Let $\bar\Omega^+ = \cup_{i=1}^N \bar\Omega^+_i$ be a non-trivial
  minimizer of $\bar E$, where $\bar\Omega^+_i$ are the disjoint
  connected components of $\bar\Omega^+$, and let $N \geq 2$. Then
  there exists $\alpha > 0$ such that 
  \begin{eqnarray}
    \label{eq:al}
    \mathrm{dist} (\bar\Omega^+_i, \bar \Omega^+_j) >
    \ep^{-\alpha} \qquad \forall i \not = j,
  \end{eqnarray}
  for $\ep \ll 1$.
\end{lemma}

\begin{proof}
  Consider the second variation of $\bar E$ with respect to the
  perturbation, in which the boundary of each connected component
  $\bar \Omega^+_i$ is expanded uniformly by a distance $a c_i$ in the
  normal direction, i.e., we have $\rho(\bar x) = c_i$ for all $\bar x
  \in \partial \bar\Omega^+_i$. By (\ref{eq:Ebvar3}), we have
  \begin{eqnarray}
    \label{eq:57}
    \left. {d^2 \bar E  (\bar \Omega^+_a) \over da^2}
  \right|_{a = 0}  = |\ln \ep|^{-1} \sum_{i,j} Q_{ij} c_i c_j,
  \end{eqnarray}
  where the coefficients $Q_{ij}$ of the quadratic form $Q$ can be
  estimated as
  \begin{eqnarray}
    \label{eq:Qii}
    Q_{ii} & = & - \int_{\partial \bar\Omega^+_i} K^2(\bar x) \, d
    \mathcal H^1(\bar x) + {2 \over 3 \pi} |\partial
    \bar\Omega^+_i|^2 + o(1),  \\ 
  \end{eqnarray}
  where we took into account that by (\ref{eq:vbeq}) and Gauss's
  theorem $\int_{\partial \bar\Omega^+_i} \nu(\bar x) \cdot \nabla
  v(\bar x) \, d \mathcal H^1(\bar x) = - |\ln \ep|^{-1} |\bar
  \Omega^+_i| + O(\ep^{2/3} |\ln \ep|^{-2/3})$ and used the expansion
  in (\ref{eq:Gest}) together with Lemmas \ref{l:K}, \ref{l:v} and
  \ref{l:EcC}, for $\ep \ll 1$. Furthermore, since by convexity of
  $\bar\Omega^+_i$ (see Lemma \ref{l:K}) the boundary of each
  $\bar\Omega^+_i$ is a closed curve, by Cauchy-Schwarz inequality we
  have
  \begin{eqnarray}
    \label{eq:59}
    4 \pi^2 = \left( \int_{\partial\bar\Omega^+_i} K(\bar x) \, d
      \mathcal  H^1(\bar x) 
    \right)^2 \leq |\partial \bar\Omega^+_i| \int_{\partial
      \bar\Omega^+_i} K^2(\bar x) \, d  \mathcal H^1(\bar x).
  \end{eqnarray}
  Therefore, the diagonal elements of $Q$ can be further estimated as
  \begin{eqnarray}
    \label{eq:58}
    Q_{ii} \leq {2 \over 3 \pi}  |\partial \bar\Omega^+_i|^2 - {4
      \pi^2 \over  |\partial \bar\Omega^+_i|} + o(1).
  \end{eqnarray}
  On the other hand, define $\alpha = |\ln \ep|^{-1} \ln
  (\mathrm{dist} (\bar \Omega^+_i, \bar \Omega^+_j))$, and suppose, to
  the contrary of the statement of the proposition, that $\alpha$ is
  sufficiently small for some pair of indices for a sequence of $\ep
  \to 0$. Then, with the help of Lemma \ref{l:dist} and
  (\ref{eq:Gest}) we can estimate
  \begin{eqnarray}
    \label{eq:60}
    Q_{ij} = {2 \over 3 \pi} (1 - 3 \alpha) \, |\partial
    \bar\Omega^+_i| \, |\partial \bar\Omega^+_j| + o(1). 
  \end{eqnarray}

  Now, for the index pair $(i,j)$ above let us choose $c_i = |\partial
  \bar\Omega^+_j|$, $c_j = -|\partial \bar\Omega^+_i|$, and let us set
  $c_k = 0$ for all other indices $k$. A simple calculation of the sum
  in (\ref{eq:57}) then shows that for this choice of $c$'s we have
  \begin{eqnarray}
    \label{eq:111}
    |\ln \ep| \left. {d^2 \bar E (\bar \Omega^+_a) \over da^2}
    \right|_{a = 0} 
    \leq {4  |\partial
      \bar\Omega^+_i|^2 |\partial
      \bar\Omega^+_j|^2 \over \pi} \left( \alpha - {\pi^3 \over
        |\partial \bar\Omega^+_i|^3} - {\pi^3 \over |\partial
      \bar\Omega^+_j|^3} \right) + o(1),
  \end{eqnarray}
  where we took into account Lemma \ref{l:EcC}. This expression is
  negative for $\ep \ll 1$, if
  \begin{eqnarray}
    \label{eq:62}
    \alpha < 2 \pi^3 \min \{ |\partial \bar\Omega^+_i|^{-3}, |\partial
    \bar\Omega^+_j|^{-3}  \},
  \end{eqnarray}
  which, in view of Lemma \ref{l:EcC}, contradicts minimality of $\bar
  E$ for small enough $\ep$.
\end{proof}

Let us also point out that the proof of Proposition \ref{l:al} gives a
universal lower bound for the perimeter of the connected portions of
the minimizers. Indeed, the quadratic form $Q$ has a negative
eigenvalue, if ${2 \over 3 \pi} |\partial \bar \Omega^+_i|^2 - 4 \pi^2
|\partial \bar \Omega^+_i|^{-1} < 0$ and $\ep \ll 1$, which implies
the following result:
\begin{proposition}
  \label{p:perlower}
  Let $\bar\Omega^+ = \cup_{i=1}^N \bar\Omega^+_i$ be a non-trivial
  minimizer of $\bar E$, where $\bar\Omega^+_i$ are the disjoint
  connected components of $\bar\Omega^+$. Then, for every $\delta > 0$ 
  \begin{eqnarray}
    \label{eq:permin}
    |\partial \bar \Omega^+_i| \geq \pi \sqrt[3]{6} - \delta,
  \end{eqnarray}
  for $\ep \ll 1$.
\end{proposition}

\noindent Note that this criterion in the radially-symmetric case was
obtained in \cite{m:pre02,m:phd,mo1:pre96} and is also applicable to
all local minimizers (for global minimizers, a better bound will be
obtained below).  We also derive another quantitative estimate on $v$
and the geometry of $\bar\Omega^+_i$ that remains valid for local
minimizers of low energy.
\begin{proposition}
  \label{p:areaupper}
  Let $\bar\Omega^+ = \cup_{i=1}^N \bar\Omega^+_i$ be a non-trivial
  minimizer of $\bar E$, where $\bar\Omega^+_i$ are the disjoint
  connected components of $\bar\Omega^+$. Then
  \begin{eqnarray}
    \label{eq:areamax}
    0 < v < {\bar \delta \over 2 \kappa^2}, \qquad
    |\bar \Omega^+_i| < {3 \pi \bar \delta \over \kappa^2},  
  \end{eqnarray}
  for $\ep \ll 1$.
\end{proposition}

\begin{proof}
  Let $\bar x \in \partial \bar\Omega^+_i$. Then, using Lemma
  \ref{l:K}, (\ref{eq:ELsh}), (\ref{eq:Gest}), and positivity of $G$,
  for some $c > 0$ independent of $\ep \ll 1$ we obtain
  \begin{eqnarray}
    \label{eq:71}
    0 <  c \leq K(\bar x) = 2 \bar \delta \kappa^{-2} - 4 v(\bar x) 
    \nonumber \\
    \leq  \left( 2 \bar \delta \kappa^{-2} - 4 |\ln \ep|^{-1}
      \int_{\bar\Omega^+_i} G(\ep^{1/3} |\ln \ep|^{-1/3} (\bar x -
      \bar y)) d \bar y \right) \nonumber \\ 
    \leq 2 \bar \delta \kappa^{-2} - {2 \over 3 \pi} |\bar\Omega^+_i|
    + o(1), 
  \end{eqnarray}
  which, together with (\ref{eq:osc}) and the fact that $v$ reaches
  its maximum in the closure of $\bar\Omega^+$, yields the statement.
\end{proof}

We now prove that for $\ep \ll 1$ each droplet in a minimizer is, in
fact, close to a disk. The basic idea of the proof is that because of
the logarithmic behavior of $G$ at small distances the potential $v$
inside each droplet is approximately constant. Therefore, the shape of
the droplet approximately minimizes the usual isoperimetric problem,
and the size of the droplet is determined by the balance of surface
tension and the pressure due to non-local forces inside the droplet
\cite{m:phd,m:pre02}.

If the droplet $\bar \Omega^+_i$ were exactly a disk $B_{\bar
  r_i}(\bar x_i)$ of radius $\bar r_i$ centered at $\bar x_i$, then
the potential $v$ would be given by
\begin{eqnarray}
  \label{eq:vstar}
  v^*(\bar x) = v^*_i(\bar x) + v_i(\bar x), 
\end{eqnarray}
where
\begin{eqnarray}
  \label{eq:113}
  v^*_i(\bar x) = |\ln \ep|^{-1} \sum_{\mathbf n \in \mathbb Z^2}
  v^B(|\bar x - \bar x_i - \mathbf n|,  \bar r_i, \ep^{1/3} |\ln
  \ep|^{-1/3} \kappa), 
\end{eqnarray}
with the function $v^B(\rho, r, \kappa)$ being the solution of
$-\Delta v^B + \kappa^2 v^B = \chi_{B_r(0)}$ in $\mathbb R^2$, given
explicitly in terms of the modified Bessel functions:
\begin{eqnarray}
  \label{eq:vB}
  v^B(\rho, r, \kappa) = 
  \begin{cases}
    \kappa^{-2} - \kappa^{-1} r K_1(\kappa r) I_0 ( \kappa \rho), &
    \rho \leq r, \\
    \kappa^{-1} r I_1( \kappa r) K_0 ( \kappa \rho), & \rho \geq r,
  \end{cases}
\end{eqnarray}
and
\begin{eqnarray}
  \label{eq:vi}
  v_i(\bar x) = |\ln \ep|^{-1} \sum_{j \not = i} 
  \int_{\bar\Omega^+_j} G(\ep^{1/3} |\ln \ep|^{-1/3} (\bar x - \bar
  y)) \, d  \bar y.
\end{eqnarray}
Note that in view of Lemmas \ref{l:al} and \ref{l:EcC}, and
(\ref{eq:Gest}), we have
\begin{eqnarray}
  \label{eq:val}
  |\nabla v_i| \leq C \ep^\alpha, 
\end{eqnarray}
for some $C > 0$ and $\alpha > 0$, in any disk of $O(1)$ radius
containing $\bar\Omega^+_i$ for $\ep \ll 1$. Therefore, if $\bar x
\in \partial B_{\bar r_i}(\bar x_i)$, by Taylor-expanding the Bessel
functions \cite{abramowitz} we have for any $\alpha < \tfrac13$
\begin{eqnarray}
  \label{eq:114}
  v^*(\bar x) = \bar v_i + O(\ep^\alpha), \qquad |\nabla v^*(\bar x)|
  = O(|\ln \ep|^{-1}), 
\end{eqnarray}
where the constant $\bar v_i$ is given by
\begin{eqnarray}
  \label{eq:vib}
  \bar v_i = - \tfrac12 |\ln \ep|^{-1} \bar r_i^2 
  \ln (\ep^{1/3} |\ln \ep|^{-1/3} \bar \kappa \bar r_i) \hspace{5cm}
  \nonumber \\
  + \pi |\ln \ep|^{-1} \sum_{j \not=
    i} \bar r_j^2 G (\ep^{1/3} |\ln \ep|^{-1/3} \kappa (\bar
  x_i - \bar  x_j)), \qquad \bar r_j = (\tfrac{1}{\pi}
  |\bar\Omega^+_j|)^{1/2}.  
\end{eqnarray}
In the following, we will show that $v(\bar x)$ on $\partial
\bar\Omega^+_i$ also coincides with $\bar v_i$ to $O(\ep^\alpha)$,
giving the balance of forces at the interface. We are now ready to
state our result:

\begin{proposition}
  \label{p:disk}
  Let $\bar\Omega^+ = \cup_{i=1}^N \bar\Omega^+_i$ be a non-trivial
  minimizer of $\bar E$, where $\bar\Omega^+_i$ are the disjoint
  connected components of $\bar\Omega^+$. Then there exists a constant
  $\alpha > 0$ such that for all $\ep \ll 1$:

  \begin{itemize}

 \item[(i)] For each $\bar\Omega^+_i$ there exists a point $\bar x_i
    \in \bar \Omega^+_i$, such that
    \begin{eqnarray}
      \label{eq:74}
      B_{\bar r_i-\ep^\alpha}(\bar x_i) \subset
      \bar\Omega^+_i \subset B_{\bar r_i+\ep^\alpha}(\bar x_i),    
    \end{eqnarray}
    where $\bar r_i = (\tfrac{1}{\pi} |\bar\Omega^+_i|)^{1/2}$ and
    $B_{\bar r}(\bar x)$ is a disk of radius $\bar r$ centered at
    $\bar x$;


  \item[(ii)]  The values of $\bar r_i$ satisfy
    \begin{eqnarray}
      \label{eq:bri}
      \bar r_i^{-1} - 2 \bar \delta \kappa^{-2}+ 4 \bar v_i =
      O(\ep^\alpha), 
    \end{eqnarray}
    where $\bar v_i$ are given by (\ref{eq:vib}).
 \end{itemize}
\end{proposition}

\begin{proof}
  Let us pick a point $\bar x_i' \in \bar\Omega^+_i$, then
  $\bar\Omega^+_i \subset B_{|\partial \bar \Omega^+_i|}(\bar
  x_i')$. Let us then replace $\bar\Omega^+_i$ with a disk of the same
  area centered at $\bar x_i'$. By Lemmas \ref{l:al} and \ref{l:EcC},
  the resulting set $B_{\bar r_i}(\bar x')$ still satisfies the bound
  in (\ref{eq:al}), and the change of energy $\Delta \bar E$ under
  this rearrangement can be estimated as
  \begin{eqnarray}
    \label{eq:68}
    |\ln \ep| \Delta \bar E  = 2 \sqrt{\pi} \, |\bar
    \Omega^+_i|^{1/2} - |\partial \bar\Omega^+_i| + O\left(  |\ln
      \ep|^{-1} \right), 
  \end{eqnarray}
  where we used (\ref{eq:Gest}), (\ref{eq:val}) and Lemma \ref{l:EcC}.
  Thus, the energy will decrease under this rearrangement,
  contradicting minimality of $\bar E$, unless for some $C > 0$ the
  isoperimetric deficit of $\bar\Omega^+_i$
  \begin{eqnarray}
    \label{eq:72}
    D(\bar\Omega^+_i) = {|\partial \bar\Omega^+_i| \over 2 \sqrt{\pi}
      \, |\bar\Omega^+_i|^{1/2}} - 1 \leq C  |\ln
    \ep|^{-1}  ,
  \end{eqnarray}
  for $\ep \ll 1$. Choosing $\bar x_i \in B_{|\partial \bar
    \Omega^+_i|}(\bar x_i')$ to minimize $|\bar\Omega^+_i \Delta
  B_{\bar r_i}(\bar x_i)|$, where $\bar\Omega^+_i \Delta B_{\bar
    r_i}(\bar x_i)$ denotes the symmetric difference of sets
  $\bar\Omega^+_i$ and $B_{\bar r_i}(\bar x_i)$, by the results of
  \cite{fusco08} we have $|\bar\Omega^+_i \Delta B_{\bar r_i}(\bar
  x_i)| \leq C' |\ln \ep|^{-1/2}$, and $C' > 0$ is a constant
  independent of $\ep \ll 1$. Then, by Lemma \ref{l:K} the set
  $\partial \bar\Omega^+_i$ is uniformly close to $\partial B_{\bar
    r_i}(\bar x_i)$, and $\bar x_i \in \bar\Omega^+_i$ (if not, then
  by convexity of $\bar\Omega^+_i$ we would have $|\bar\Omega^+_i
  \Delta B_{\bar r_i}(\bar x_i)| \geq \tfrac12 |B_{\bar r_i}(\bar
  x_i)|$), giving (i) to $o(1)$.

  To obtain the $O(\ep^\alpha)$ bound in (i), let $\rho: \partial
  B_{\bar r_i}(\bar x_i) \to \mathbb R$ be the signed distance from a
  given point on $\partial B_{\bar r_i}(\bar x_i)$ to $\partial
  \bar\Omega^+_i$ along the outward normal to $\partial B_{\bar
    r_i}(\bar x_i)$. Note that by convexity of $\bar\Omega^+_i$ the
  function $\rho$ defines a one-to-one map between $\partial
  \bar\Omega^+_i$ and $\partial B_{\bar r_i}(\bar x_i)$. Furthermore,
  if $||\rho||_{L^\infty(\partial B_{\bar r_i}(\bar x_i))} = \delta$,
  we have $||\nabla \rho||_{L^\infty(\partial B_{\bar r_i}(\bar x_i))}
  \leq C \delta^{1/2}$ for some $C > 0$ and $\ep \ll 1$ in view of
  Lemma \ref{l:K}, and $\delta \to 0$, as $\ep \to 0$. Also, by
  Corollary \ref{c:c3} we have $\rho \in C^3(\partial B_{\bar
    r_i}(\bar x_i))$.

  Treating $\bar \Omega^+$ as a perturbation of the set $\bar\Omega^*
  = B_{\bar r_i}(\bar x_i) \cup (\bar \Omega^+ \backslash \bar
  \Omega^+_i)$ and expanding as in Lemma \ref{l:vars}, we can write
  \begin{eqnarray}
    \label{eq:Estar}
    |\ln \ep| (\bar E(\bar\Omega^+) - \bar E(\bar \Omega^*)) =
    \int_{\partial B_{\bar r_i}(\bar x_i)} \left( \bar r_i^{-1} - 2
      \bar \delta \kappa^{-2} + 4 v^*(\bar x) \right) \rho(\bar x) \,
    d \mathcal H^{1}(\bar x) \nonumber \\ 
    + \frac12  \int_{\partial B_{\bar r_i}(\bar x_i)} \left( |\nabla
      \rho(\bar x)|^2 + 4 \nu(\bar x) \cdot \nabla v^*(\bar x) \,
      \rho^2(\bar x) \right) \, d \mathcal H^1(\bar x)  \nonumber \\ 
    + {1 \over 2 \bar r_i}  \int_{\partial B_{\bar r_i}(\bar
      x_i)}  (4 v^*(\bar x) - 2 \bar \delta \kappa^{-2}) 
    \rho^2(\bar x) \, d \mathcal H^1(\bar x) \nonumber \\ 
    + {2 \over |\ln \ep|} \int_{\partial B_{\bar r_i}(\bar x_i)}
    \int_{\partial B_{\bar r_i}(\bar x_i)} G(\ep^{1/3} |\ln
    \ep|^{-1/3} (\bar x - \bar y) \, \rho(\bar x) \rho(\bar y) \, d
    \mathcal H^1(\bar x) d \mathcal H^1(\bar y) \nonumber \\ 
    + O(\delta^{2+\alpha}), 
  \end{eqnarray}
  for any $\alpha \in (0, 1)$. Moreover, in view of Lemmas \ref{l:v}
  and \ref{l:K}, the error term in (\ref{eq:Estar}) is uniform in $\ep
  \ll 1$.

  On the other hand, since $|\bar\Omega^+_i| = |B_{\bar r_i}(\bar
  x_i)|$, we have
  \begin{eqnarray}
    \label{eq:78}
    0 = \int_{\partial B_{\bar r_i}(\bar x_i)} \int_0^{\rho(\bar x)} (1 +
    \bar r_i^{-1} r) \, dr \, d \mathcal H^1(\bar x) \hspace{3cm}
    \nonumber \\  
    = \int_{\partial B_{\bar r_i}(\bar x_i)} \rho(\bar x) \, d \mathcal
    H^1(\bar x) + {1 \over 2 \bar r_i} \int_{\partial B_{\bar
        r_i}(\bar x_i)} \rho^2 (\bar x) \, d \mathcal
    H^1(\bar x). 
  \end{eqnarray}
  Therefore, using the estimate in (\ref{eq:114}) we can rewrite
  (\ref{eq:Estar}) as
  \begin{eqnarray}
    \label{eq:Estar2}
    |\ln \ep| (\bar E(\bar\Omega^+)-  \bar E(\bar\Omega^*)) =
    \frac12 \int_{\partial B_{\bar r_i}(\bar x_i)} \left( |\nabla
      \rho|^2 -  \bar r_i^{-2}
      \rho^2(\bar x) \right) d \mathcal H^1(\bar x) 
    \nonumber \\ 
    + {2 \over |\ln \ep|} \int_{\partial B_{\bar r_i}(\bar x_i)}
    \int_{\partial B_{\bar r_i}(\bar x_i)} G(\ep^{1/3} |\ln
    \ep|^{-1/3} (\bar x - \bar y) \, \rho(\bar x) \rho(\bar y) \, d
    \mathcal H^1(\bar x) d \mathcal H^1(\bar y) \nonumber \\ 
    + O(|\ln \ep|^{-1} ||\rho||^2_{L^2_{B_{\bar r_i}(\bar x_i)}})
    + O(\ep^\alpha ||\rho||_{L^2_{B_{\bar r_i}(\bar x_i)}})
    + O(\delta^\alpha ||\rho||_{H^1_{B_{\bar r_i}(\bar x_i)}}^2),
 \end{eqnarray}
 where we took into account that $\delta \leq C ||\rho||_{H^1_{B_{\bar
       r_i}(\bar x_i)}}$ for some $C > 0$.  Further estimating the
 double integral in (\ref{eq:Estar2}), using
   \begin{eqnarray}
    \label{eq:76}
    \left| \int_{\partial B_{\bar r_i}(\bar x_i)}
      \left(  {1 \over |\ln \ep|}  G(\ep^{1/3} |\ln \ep|^{-1/3} (\bar
        x - \bar y)) - {1 \over 6 \pi} \right) \, \rho(\bar y) 
      \, d \mathcal H^1(\bar y) \right|
    \nonumber \\
    \leq {C \delta \ln |\ln \ep| \over |\ln \ep|},
  \end{eqnarray}
  we have
  \begin{eqnarray}
    \label{eq:2}
    |\ln \ep| (\bar E(\bar\Omega^+)-  \bar E(\bar\Omega^*)) =
    \nonumber \\ 
    \frac12 \int_{\partial B_{\bar r_i}(\bar x_i)} \left( |\nabla
      \rho|^2 -  \bar r_i^{-2}
      \rho^2(\bar x) \right) d \mathcal H^1(\bar x) 
    + {1 \over 3 \pi}
    \left(  \int_{\partial B_{\bar r_i}(\bar x_i)} \rho(\bar x) \, d
      \mathcal H^1(\bar x) \right)^2  \nonumber \\ 
    + O(\ep^\alpha ||\rho||_{L^2_{B_{\bar r_i}(\bar x_i)}})
    + o( ||\rho||_{L^2_{B_{\bar r_i}(\bar x_i)}}^2) 
    + o( ||\rho||_{H^1_{B_{\bar r_i}(\bar x_i)}}^2).
 \end{eqnarray}


 Now, write $\rho$ as $\rho = \rho_0 + \rho_1 + \rho_2$, where $\rho_0
 = \tfrac{1}{2 \pi \bar r_i} \int_{\partial B_{\bar r_i}(\bar x_i)}
 \rho(\bar x) \, d \mathcal H^1(\bar x)$, $\rho_1(\bar x) = {\bar x -
   \bar x_i \over |\bar x - \bar x_i|} \cdot b$, for some vector $b
 \in \mathbb R^2$, and $\rho_2$ orthogonal to $\rho_0$ and $\rho_1$ in
 $L^2(\partial B_{\bar r_i}(\bar x_i))$. By (\ref{eq:78}) we have
 $|\rho_0| = O(||\rho||_{L^2(\partial B_{\bar r_i}(\bar x_i))}^2)$,
 which is, therefore, negligibly small compared to $|b|$ and
 $||\rho_2||_{L^2(\partial B_{\bar r_i}(\bar x_i))}$ in all the
 arguments below. Then, using Poincar\'e's inequality, we find that
  \begin{eqnarray}
    \label{eq:79}
    \bar E(\bar\Omega^+) \geq \bar E(\bar\Omega^*) \hspace{7cm} 
    \nonumber \\ 
    + \frac14
    \int_{\partial B_{\bar  r_i}(\bar x_i)} |\nabla 
    \rho_2|^2(\bar x) \, d \mathcal H^1(\bar x) 
    - C \ep^{\alpha} ||\rho||_{L^2_{B_{\bar r_i}(\bar
        x_i)}} - c |b|^2,  
  \end{eqnarray}
  for some $C > 0$ and $0 < c \ll 1$, whenever $\ep \ll 1$. This
  implies that
  \begin{eqnarray}
    \label{eq:81}
    ||\rho_2||_{H^1_{B_{\bar r_i}(\bar x_i)}}^2 \leq 4 C
    \ep^\alpha ||\rho||_{L^2_{B_{\bar 
          r_i}(\bar x_i)}} + 4 c |b|^2, 
  \end{eqnarray}
  for $\ep \ll 1$, otherwise replacing $\bar\Omega^+_i$ with $B_{\bar
    r_i}(\bar x_i)$ lowers the energy. On the other hand, we also have
  $|b| = O(||\rho_2||_{H^1_{B_{\bar r_i}(\bar x_i)}})$. If not, then
  $\partial \bar\Omega^+_i$ will be $o(|b|)$ close to $\partial
  B_{\bar r_i}(\bar x_i + b)$ for $\ep \ll 1$. This, however,
  contradicts the choice of $\bar x_i$ to minimize $|\bar\Omega^+_i
  \Delta B_{\bar r_i}(\bar x_i)|$. Therefore, we have
  \begin{eqnarray}
    \label{eq:11}
    ||\rho||_{H^1_{B_{\bar r_i}(\bar x_i)}}^2 \leq C \ep^{\alpha}
    ||\rho||_{H^1_{B_{\bar r_i}(\bar x_i)}} + c
    ||\rho||_{H^1_{B_{\bar r_i}(\bar x_i)}}^2,
  \end{eqnarray}
  for some $C > 0$ and $0 < c \ll 1$, implying $||\rho||_{H^1_{B_{\bar
        r_i}(\bar x_i)}} = O(\ep^\alpha)$ and, hence, $\delta =
  O(\ep^\alpha)$. This gives part (i) of the statement of the
  proposition.

  Finally, to prove part (ii) of the statement, let $\bar\Omega^+_a$
  be obtained from $\bar \Omega^+$ by expanding $\bar \Omega^+_i$ by
  an amount $a > 0$, i.e., let us change $\rho(\bar x) \to \rho(\bar
  x) + a$ for every $\bar x \in \partial B_{\bar r_i}(\bar x_i)$. By
  (\ref{eq:Estar}), the change of energy can be estimated as
  \begin{eqnarray}
    \label{eq:77}
    |\ln \ep| (\bar E(\bar\Omega^+_a) - \bar E(\bar\Omega^+)) =  2
    \pi a \bar r_i (\bar r_i^{-1}  - 2 \bar\delta \kappa^{-2} + 4 \bar
    v_i) + O(a \delta)  + O(a^2),
  \end{eqnarray}
  where we took into account (\ref{eq:114}).  Then, since $\bar
  \Omega^+$ is a minimizer, the right-hand side of (\ref{eq:77})
  should vanish to $O(a)$. Therefore, by previous result we obtain the
  statement.
\end{proof}

Also, from the proof of Proposition \ref{p:disk} we obtain the
following universal lower bound on $|\bar \Omega^+_i|$:
\begin{proposition}
  \label{p:amin}
  Let $\bar\Omega^+ = \cup_{i=1}^N \bar\Omega^+_i$ be a non-trivial
  minimizer of $\bar E$, where $\bar\Omega^+_i$ are the disjoint
  connected components of $\bar\Omega^+$. Then, for every $\delta > 0$
  \begin{eqnarray}
    \label{eq:areamin}
    |\bar \Omega^+_i| \geq \pi \sqrt[3]{9} - \delta,
  \end{eqnarray}
  for $\ep \ll 1$.
\end{proposition}

\begin{proof}
  By Proposition \ref{p:disk} each value of $\bar r_i =
  (\tfrac{1}{\pi} |\bar\Omega^+_i|)^{1/2}$ satisfies (\ref{eq:bri})
  and, hence, is close to a critical point of $V_{v_i(\bar x_i)}$
  defined in (\ref{eq:Vbri}), which can be seen from (\ref{eq:vib}) by
  Taylor expansion. Furthermore, we must have $2 \bar v_i(\bar x_i) -
  \bar \delta \kappa^{-2} \leq -\tfrac12 \sqrt[3]{9} + o(1)$, so that
  $V_{v_i(\bar x_i)}(\bar r_i) \leq o(1)$ for $\ep \ll 1$, otherwise,
  arguing as in Lemma \ref{l:EcC}, we can reduce the energy by
  removing $\bar\Omega^+_i$ from $\bar\Omega^+$. Therefore, $\bar r_i$
  should be close to the positive minimum of $V_{v_i(\bar x_i)}$. By
  inspection, in this situation $\bar r_i \geq \sqrt[3]{3} - \delta$,
  for any $\delta > 0$ (see Fig. \ref{fig:fa}), provided $\ep$ is
  sufficiently small, hence, the claim.
\end{proof}

The results of Proposition \ref{p:disk} just obtained immediately
allow to establish an asymptotic equivalence of the energy $\bar E$
and the reduced energy $\bar E_N$ on the minimizers for $\ep \ll 1$.

\begin{proposition}
  \label{p:EEE}
  Let $\bar\Omega^+ = \cup_{i=1}^N \bar\Omega^+_i$ be a non-trivial
  minimizer of $\bar E$, where $\bar\Omega^+_i$ are the disjoint
  connected components of $\bar\Omega^+$, and let $\bar r_i$ and $\bar
  x_i$ be as in Proposition \ref{p:disk}. Then
  \begin{eqnarray}
    \label{eq:EEN}
    \min \bar E = O(1), \qquad \min \bar E = \min
    \bar E_N + O(\ep^\alpha),    
  \end{eqnarray}
  for some $\alpha > 0$ independent of $\ep \ll 1$. 
\end{proposition}

\begin{proof}
  The first equation in (\ref{eq:EEN}) is a direct consequence of the
  definition of $\bar E$ in (\ref{eq:Ebar}), according to which $0
  \leq \tfrac12 \bar \delta^2 \kappa^{-2} + \min \bar E \leq
  \ep^{-4/3} |\ln \ep|^{2/3} E[-1] = \tfrac12 \bar \delta^2
  \kappa^{-2}$.  The upper bound in the second equation follows by
  choosing a trial function for $\bar E$ in the form of disks of
  radius $\bar r_i$ centered at $\bar x_i$ which minimize $\bar E_N$
  and taking into consideration Lemmas \ref{l:EcC} and \ref{l:al} and
  (\ref{eq:Gest}). On the other hand, by Proposition \ref{p:disk}(i),
  we have $\bar \Omega^+_i \supset B_{\bar r_i - \ep^\alpha}(\bar
  x_i)$ for $\ep \ll 1$, hence, $|\partial \bar\Omega^+_i| > 2 \pi
  (\bar r_i - \ep^\alpha)$ and $|\bar\Omega^+_i| > \pi (\bar r_i -
  \ep^\alpha)^2$. This controls from below all the terms of $\bar E$,
  except the one involving $\bar \delta$, by the corresponding terms
  of $\bar E_N$. The latter, however, is controlled by the second
  inclusion in Proposition \ref{p:disk}(i).
\end{proof}

To summarize, for $0 < \ep \ll 1$ the non-trivial minimizers of $\bar
E$ have the form of well-separated nearly circular droplets. In fact,
from Proposition \ref{p:EEE} one should expect that the
droplet-droplet interaction part of the energy, which is given by the
last term in the expression (\ref{eq:Enb}) for $\bar E_N$, should be
close to the minimum for fixed droplet sizes. Proving this, however,
generally requires information about coercivity of the interaction
energy, which becomes difficult to establish when $N \gg 1$, the
asymptotic case of interest. Nevertheless, with the help of Lemma
\ref{l:v} we can prove that in the original scaling the droplets stay
away from each other a distance $O(\ep^\beta)$ in $\Omega$, with an
arbitrary $\beta > 0$ for $\ep \ll 1$, i.e., that the statement of
Lemma \ref{l:al} actually holds for any $\alpha \in (0, \tfrac13)$,
provided that $\ep$ is small enough.

\begin{proposition}
  \label{p:beta}
  Let $\bar\Omega^+ = \cup_{i=1}^N \bar\Omega^+_i$ be a non-trivial
  minimizer of $\bar E$, where $\bar\Omega^+_i$ are the disjoint
  connected components of $\bar\Omega^+$, and let $\bar x_i$ be as in
  Proposition \ref{p:disk}. Then, for any $\alpha \in (0, \tfrac13)$
  we have $|\bar x_i - \bar x_j| > \ep^{-\alpha}$, for all $i \not=
  j$, as long as $\ep \ll 1$.
\end{proposition}

\begin{proof}
  First of all, note that by Lemma \ref{l:al} the statement of the
  Proposition holds for some $\alpha > 0$. To prove that $\alpha$
  could be chosen arbitrarily close to $\tfrac13$, suppose that, to
  the contrary, there exists a sequence of $\ep \to 0$ and a pair of
  indices $(i, j)$, depending on $\ep$, such that $|\bar x_i - \bar
  x_j| \leq \ep^{-\alpha}$ with some $0 < \alpha < \tfrac13$. Let us
  denote by $I_1$ the set of indices of those droplets whose centers
  are contained in a disk $B_1$ centered at $\bar x_0 = \tfrac12 (\bar
  x_i + \bar x_j)$ with radius $\ep^{-\alpha}$. By assumption we have
  $|I_1| \geq 2$, where $|\cdot|$ denotes the counting measure. Also,
  we have $|I_1| < M$ for some $M \in \mathbb N$ independent of $\ep
  \ll 1$. Indeed, by Lemmas \ref{l:EcC} and \ref{l:v}, and by
  (\ref{eq:Gest}) we have for some $c > 0$
  \begin{eqnarray}
    \label{eq:93}
    C \geq v(\bar x_0) \geq |\ln \ep|^{-1} \sum_{k \in I_1}
    \int_{\bar\Omega^+_k} G(\ep^{1/3} |\ln \ep|^{-1/3} (\bar x_i -
    \bar y)) \, d\bar y \nonumber \\ 
    \geq \left(\tfrac13 -\alpha  + o(1) \right) c |I_1|,
  \end{eqnarray}
  for $\ep \ll 1$.

  Now, fix $\sigma > 0$ sufficiently small independently of $\ep$, and
  consider a sequence of nested disks $B_k$ of radii $\ep^{-\alpha (1
    + k \sigma)}$ centered at $\bar x_0$. By repeating the argument
  above, we also have $|I_M| \leq M$, as long as $\ep \ll 1$, where
  $|I_k|$ is the counting measure of the set $I_k$ of indices such
  that $\bar x_l \in B_k$ for all $l \in I_k$. Therefore, in view of
  the fact that $|I_1| > 1$, we must have $|I_{k+1}| - |I_k| = 0$ for
  some $1 \leq k \leq M-1$, implying that $B_{k+1} \backslash B_k \cap
  \bar \Omega^+ = \varnothing$. Thus, there exists a cluster of
  droplets, whose indices are denoted by $I_k$, which are within
  $O(\ep^{-\alpha (1 + k \sigma)})$ distance of $\bar x_0$ and are
  separated from all other droplets by $O(\ep^{-\alpha (1 + \sigma + k
    \sigma)})$ distance.

  Let us show that this contradicts the minimality of $\bar E$ for
  small enough $\ep$. Indeed, let us displace the droplets in $B_k$ to
  the new locations $\bar x'_l = \bar x_l + \lambda (\bar x_l - \bar
  x_i)$, with $l \in I_k$, which represents a dilation of $B_k$ by a
  factor of $1 + \lambda$ relative to $\bar x_i$, keeping all $\bar
  r_i$ fixed. For $0 < \lambda \ll 1$ the resulting change $\Delta
  \bar E$ of energy satisfies
  \begin{eqnarray}
    \label{eq:6}
    |\ln \ep|^2 \Delta \bar E \leq - c \lambda  + C \lambda \ep^{\sigma
      \alpha} |\ln  \ep|,  
  \end{eqnarray}
  for some $C, c > 0$ independent of $\ep \ll 1$, where we used Lemmas
  \ref{l:bbounds} and \ref{l:EcC}, and the estimate (\ref{eq:Gest}),
  arguing as in the derivation of (\ref{eq:val}). Thus, the considered
  rearrangement lowers the energy.
\end{proof}

As a simple corollary to this result, we actually have the following
universal ($\bar\delta$-independent) upper bound on $|\bar
\Omega^+_i|$ and, hence, on $\bar r_i$:
\begin{corollary}
  Let $\bar\Omega^+ = \cup_{i=1}^N \bar\Omega^+_i$ be a non-trivial
  minimizer of $\bar E$, where $\bar\Omega^+_i$ are the disjoint
  connected components of $\bar\Omega^+$. Then, for any $\delta > 0$
 \begin{eqnarray}
    \label{eq:25}
    |\bar \Omega^+_i| \leq \pi \left( 12 (\sqrt{2} - 1)
    \right)^{2/3} + \delta,
  \end{eqnarray}
  when $\ep$ is sufficiently small.
\end{corollary}

\begin{proof}
  If $|\bar \Omega^+_i|$ is bigger, split $\bar\Omega^+_i$ into two
  disks of equal area and move them apart a distance $d =
  \ep^{-\beta}$, with $0 < \beta < \alpha < \tfrac13$. Arguing as
  before, the energy change $\Delta \bar E$ upon this manipulation is
  given by
  \begin{eqnarray}
    \label{eq:26}
    |\ln \ep| \, \Delta \bar E \leq 2 (\sqrt{2} - 1) \sqrt{\pi} \,
    |\bar \Omega^+_i|^{1/2} - \frac{\beta}{2 \pi} |\bar
    \Omega^+_i|^2 + o(1).
  \end{eqnarray}
  In view of the arbitrary closeness of $\beta$ to $\tfrac13$, the
  energy change is, therefore, negative for $\ep \ll 1$.
\end{proof}

Finally, we note that the argument of Proposition \ref{p:beta} still
holds for local minimizers of low energy, the result can is obtained
by sending $\lambda \to 0$ in the proof.

\subsection{Limiting behavior}
\label{sec:lim}

We now investigate the limiting behavior of the minimizers of $\bar E$
as $\ep \to 0$, with $\bar \delta > \tfrac12 \sqrt[3]{9} \, \kappa^2$
fixed, i.e., the situation in which minimizers are non-trivial. As the
value of $\ep$ is decreased, the number of droplets in a minimizer are
expected to grow. What we will show below is that in the limit $\ep
\to 0$ the droplet sizes become asymptotically the same, and that the
droplets become uniformly distributed throughout $\Omega$.

Let us first study the behavior of each droplet as $\ep \to 0$. We
have the following result. 

\begin{proposition}
  \label{p:ri}
  Let $\bar\Omega^+ = \cup_{i=1}^N \bar\Omega^+_i$ be a non-trivial
  minimizer of $\bar E$, where $\bar\Omega^+_i$ are the disjoint
  connected components of $\bar\Omega^+$, and let $\bar r_i$ be as in
  Proposition \ref{p:disk}. Then $\bar r_i \to \sqrt[3]{3}$ uniformly
  as $\ep \to 0$. 
\end{proposition}

\begin{proof}
  First of all, by Proposition \ref{p:amin} we already know that $\bar
  r_i \geq \sqrt[3]{3} - \delta$ for any $\delta > 0$, provided that
  $\ep \ll 1$. Let us prove that the matching upper bound also holds
  for $\ep \ll 1$. Indeed, for any $\beta \in (0, \tfrac13)$ let
  $B_{\ep^{-\beta}}(\bar x_i) \in \bar\Omega$ be a disk of radius
  $\ep^{-\beta}$ centered at $\bar x_i$ defined in Proposition
  \ref{p:disk}, and consider $\bar\Omega_\beta = \bar \Omega
  \backslash \cup_{i = 1}^N \overline B_{\ep^{-\beta}} (\bar
  x_i)$. Note that by Proposition \ref{p:beta} the disks
  $B_{\ep^{-\beta}} (\bar x_i)$ do not intersect for $\ep \ll 1$. In
  fact, by Proposition \ref{p:beta}, for any $\alpha \in (\beta,
  \tfrac13)$ we have $\mathrm{dist} \, (B_{\ep^{-\beta}} (\bar x_i),
  B_{\ep^{-\beta}} (\bar x_j)) > \ep^{-\alpha}$ for $\ep \ll 1$.

  Let us show that the minimum of $v$ defined in (\ref{eq:v}) is
  attained in $\bar\Omega_\beta$ for $\ep \ll 1$. Let $\bar x$ be such
  that $v(\bar x) = \min$ and let $\bar x_i$ be the center of a
  droplet which is closest to $\bar x$. Recalling the definition in
  (\ref{eq:vi}) and Proposition \ref{p:disk}, we can write
  \begin{eqnarray}
    \label{eq:30}
    v(\bar x) = v_i(\bar x) - \frac{\bar r_i^2}{2 |\ln \ep|} \ln
    (\ep^{1/3} (1 + |\bar x - \bar x_i|)) + o(1),
  \end{eqnarray}
  where we used (\ref{eq:Gest}). In particular, for any $\delta > 0$
  we have $v(\bar x) > v_i(\bar x_i) + \tfrac16 \bar r_i^2 (1 - 3
  \beta) - \delta$, if $|\bar x - \bar x_i| \leq \ep^{-\beta}$ and
  $\ep \ll 1$, in view of (\ref{eq:val}), where according to
  Proposition \ref{p:beta}, we can use $\alpha$ defined above,
  whenever $\ep \ll 1$. On the other hand, choosing $\gamma \in
  (\beta, \alpha)$ and picking any $\bar x'$ such that $|\bar x' -
  \bar x_i| = \ep^{-\gamma}$, we see that for any $\delta > 0$ we have
  $v(\bar x') < v_i(\bar x_i) + \tfrac16 \bar r_i^2 ( 1 - 3 \gamma) +
  \delta$ for $\ep \ll 1$. However, with $\delta$ sufficiently small
  this implies that $v(\bar x') < v(\bar x)$ for small enough $\ep$,
  contradicting minimality of $v$ at $\bar x$.

  Now, we demonstrate that $v(\bar x) > \tfrac12 \bar \delta
  \kappa^{-2} - \tfrac14 \sqrt[3]{9} - \delta$, for any $\delta > 0$,
  provided that $\ep \ll 1$. Indeed, suppose the opposite inequality
  holds for some $\delta > 0$ and a sequence of $\ep \to 0$. Then,
  inserting a new droplet in the form of a disk of radius $\bar r =
  O(1)$ centered at $\bar x$ results in the change $\Delta \bar E$ of
  energy
  \begin{eqnarray}
    \label{eq:ins}
    |\ln \ep| \Delta \bar E = V_{v(\bar x)}(\bar r) + o(1),
  \end{eqnarray}
  where $V$ is given by (\ref{eq:Vbri}), and we used (\ref{eq:Gest})
  and (\ref{eq:val}). Since by assumption $2 v(\bar x) - \bar \delta
  \kappa^{-2} < \tfrac12 \sqrt[3]{9}$, it is easy to verify that
  $V_{v(\bar x)}$ attains a minimum at some $\bar r = \bar r_0 >
  \sqrt[3]{3}$, with $V_{v(\bar x)}(\bar r_0) < 0$. Therefore,
  inserting a droplet with radius $\bar r_0$ and center at $\bar x$
  would reduce energy for some $\ep \ll 1$, contradicting minimality of
  $\bar E$.

  This, in turn, implies that $v_i(\bar x_i) > \tfrac12 \bar \delta
  \kappa^{-2} - \tfrac14 \sqrt[3]{9} - \delta$ for all $i$. Indeed,
  since $\bar x \in \bar\Omega_\beta$, from (\ref{eq:30}) we have
  $v_i(\bar x') > \tfrac12 \bar \delta \kappa^{-2} - \tfrac14
  \sqrt[3]{9} - \tfrac{1}{6 \pi} (1 - 3 \beta) + o(1)$, for any $\bar
  x' \in \partial B_{\ep^{-\beta}}(\bar x_i)$, for $\ep \ll 1$. On the
  other hand, by (\ref{eq:val}) the same inequality holds for $\bar x'
  = \bar x_i$. The estimate then follows in view of arbitrariness of
  $\beta < \tfrac13$.

  Finally, since $\bar v_i(\bar x_i) > \tfrac12 \bar \delta
  \kappa^{-2} - \tfrac14 \sqrt[3]{9} - \delta$ when $\ep \ll 1$, and
  by Proposition \ref{p:disk}(ii) the values of $\bar r_i$ are close
  to the minimizers of $V_{v_i(\bar x_i)}(\bar r)$ for $\bar r >
  \sqrt[3]{3} - \delta$, by direct inspection we have $\bar r_i <
  \sqrt[3]{3} + \delta$ as well, for any $\delta > 0$ and $\ep \ll 1$.
\end{proof}

Let us point out that by (\ref{eq:bri}) the uniform convergence of the
droplet radii in Proposition \ref{p:ri} also implies uniform
convergence of $v_i(\bar x_i)$ to a space-independent constant as $\ep
\to 0$:
\begin{eqnarray}
  \label{eq:vlim}
  v_i(\bar x_i) \to \frac{1}{2 \kappa^2} \left( \bar \delta -
    \frac{\sqrt[3]{9}}{2} \,  \kappa^2  \right). 
\end{eqnarray}
In fact, from the proof of Proposition \ref{p:ri} we can conclude that
$v$ stays close to the constant in (\ref{eq:vlim}) in $\bar
\Omega_\beta$ for $\beta$ arbitrarily close to $\tfrac13$, provided
$\ep \ll 1$. This, in turn, implies that the droplets become uniformly
distributed in $\bar \Omega$, or, equivalently, in $\Omega$ as $\ep
\to 0$. Below we prove this fact, which also gives the leading order
behavior of energy in the limit.

Let us rewrite the energy $\bar E_N$ for the system of interacting
droplets, using (\ref{eq:Vbri}):
\begin{eqnarray}
  \label{eq:Eblim}
  \bar E_N = {1 \over |\ln \ep|} \sum_{i = 1}^N \left( V_{v_i(\bar
      x_i)}(\bar r_i) - 4 \pi v_i(\bar x_i) \bar r_i^2 \right)
  \nonumber \\ 
  + {4 \pi^2 \over |\ln \ep|^2} \sum_{i = 1}^{N-1} \sum_{j = i+1}^N 
  G(\ep^{1/3} |\ln \ep|^{-1/3} (\bar x_i - \bar x_j)) \bar r_i^2 \bar 
  r_j^2 + o(1).
\end{eqnarray}
To proceed, let us go back to the original scaling in $x$ and
introduce $x_i = \ep^{1/3} |\ln \ep|^{-1/3} \bar x_i$. Also, for any
$0 < \sigma \ll 1$ define (our method is reminiscent of the Ewald
summation technique)
\begin{eqnarray}
  \label{eq:75}
  G_\sigma (x) = {1 \over 4 \pi^2  \ep^{2 \sigma}} \sum_{\mathbf n \in 
    \mathbb Z^2} \int_{\mathbb R^2} e^{-{|x - y|^2 \over 2 \ep^{2 \sigma}}}
  K_0(\kappa |y - \mathbf n|) \, dy.
\end{eqnarray}
Here $G_\sigma$ is a mollified version of $G$, with Fourier transform
\begin{eqnarray}
  \label{eq:7}
  \hat G_\sigma(q) = \int_\Omega e^{i q \cdot x} G_\sigma(x) \, dx =
  {e^{-\frac12 \ep^{2 \sigma} |q|^2} \over \kappa^2 + |q|^2},
\end{eqnarray}
and which can, e.g., be estimated as
\begin{eqnarray}
  \label{eq:70}
  G_\sigma(x) = G(x) + o(\ep^{\sigma/4}), \qquad |x| > \ep^{\sigma/2},
\end{eqnarray}
and
\begin{eqnarray}
  \label{eq:92}
  G_\sigma(x) = O(\sigma |\ln \ep|), \qquad |x| <
  \ep^{\sigma}.
\end{eqnarray}
Therefore, in view of Lemmas \ref{l:EcC} and \ref{l:N} we can write
\begin{eqnarray}
  \label{eq:89}
  \bar E_N = {1 \over |\ln \ep|} \sum_{i = 1}^N \left( V_{v_i(\bar
      x_i)}(\bar r_i) - 4 \pi v_i(\bar x_i) \bar r_i^2 \right)
  \nonumber \\ 
  + {2 \pi^2 \over |\ln \ep|^2} \sum_{i = 1}^N \sum_{j =1}^N
  G_\sigma(x_i - x_j) \bar r_i^2 \bar
  r_j^2 + O(\sigma).
\end{eqnarray}

Now, let us introduce the quantity
\begin{eqnarray}
  \label{eq:rhoe}
  \rho(x) = {1 \over |\ln \ep|} \sum_{i = 1}^N \delta(x - x_i).
\end{eqnarray}
Note that by Lemma \ref{l:N} we have $\int_\Omega \rho(x) \, dx =
O(1)$. In view of Proposition \ref{p:ri} and (\ref{eq:vlim}), we can
further rewrite $\bar E_N$ as
\begin{eqnarray}
  \label{eq:94}
  \bar E_N = {1 \over |\ln \ep|} \sum_{i = 1}^N  V_{v_i(\bar
    x_i)}(\bar r_i) - {2 \pi \sqrt[3]{9} \over \kappa^2} \left( \bar
    \delta - \frac{\sqrt[3]{9}}{2} \, \kappa^2  \right) \int_\Omega
  \rho(x)  dx \nonumber \\  
  + 6 \pi^2 \sqrt[3]{3} \int_\Omega \int_\Omega \rho(x) G_\sigma (x - y)
  \rho(y) \, dx \, dy + O(\sigma).
\end{eqnarray}
In fact, by Proposition \ref{p:ri} and (\ref{eq:vlim}), the first term
in (\ref{eq:94}) goes to zero. Therefore, in terms of the Fourier
coefficients
\begin{eqnarray}
  \label{eq:95}
  \hat\rho_q = \int_\Omega e^{i q \cdot x} \rho(x) \, dx,
\end{eqnarray}
we can write
\begin{eqnarray}
  \label{eq:96}
  \bar E_N = - {2 \pi \sqrt[3]{9} \over \kappa^2} \left( \bar \delta - 
    \frac{\sqrt[3]{9}}{2} \, \kappa^2\right) \hat\rho_0 + {6 \pi^2
    \sqrt[3]{3} \over \kappa^2} \hat\rho_0^2  \nonumber \\ 
  + 6 \pi^2 \sqrt[3]{3} \sum_{q \in 2 \pi
    \mathbb Z^2 \backslash \{0\} } { e^{-\frac12 \ep^{2 \sigma} |q|^2}
    |\hat\rho_q|^2 \over \kappa^2 + |q|^2} + O(\sigma). 
\end{eqnarray}
Minimizing this with respect to $\hat\rho_0$, we obtain
\begin{eqnarray}
  \label{eq:97}
  \bar E_N \geq -{1 \over 2 \kappa^2} \left( \bar \delta -
    \frac{\sqrt[3]{9}}{2} \, \kappa^2 \right)^2 + 6 \pi^2 \sqrt[3]{3}
  \sum_{q \in 2 \pi \mathbb Z^2 \backslash \{0\} } {e^{-\frac12 \ep^{2
        \sigma} |q|^2} |\hat\rho_q|^2 \over \kappa^2 + |q|^2} +
  O(\sigma). 
\end{eqnarray}

Finally, from Lemma \ref{l:ubN} one can see that $\min \bar E_N \leq
-{1 \over 2 \kappa^2} \left( \bar \delta - \frac12 \sqrt[3]{9} \,
  \kappa^2 \right)^2 + o(1)$ for $\ep \ll 1$. Hence, in view of
arbitrariness of $\sigma$ the constant in (\ref{eq:97}) is the limit
of $\bar E_N$ and, by Proposition \ref{p:EEE}, also of $\bar E$, as
$\ep \to 0$. In addition, this implies that $\hat\rho_q \to 0$ as $\ep
\to 0$ for every $q \not = 0$. Thus, we just proved

\begin{proposition}
  \label{p:dens}
  Let $\bar \delta > \tfrac12 \sqrt[3]{9} \, \kappa^2$, and let $\rho$
  be defined in (\ref{eq:rhoe}), with $x_i = \ep^{1/3} |\ln
  \ep|^{-1/3} \bar x_i$, where $\bar x_i$ are as in Proposition
  \ref{p:disk}. Then
  \begin{eqnarray}
    \label{eq:98}
    \min \bar E \to -{1 \over 2 \kappa^2} \left( \bar \delta - 
      \frac{\sqrt[3]{9}}{2} \,  \kappa^2 \right)^2,  
  \end{eqnarray}
  and 
  \begin{eqnarray}
    \label{eq:99}
    \rho \to {1 \over 2 \pi \sqrt[3]{3}} \left( \bar \delta -
      \frac{\sqrt[3]{9}}{2} \, \kappa^2 \right)  
  \end{eqnarray}
  weakly in the sense of measures, as $\ep \to 0$.
\end{proposition}

We end by noting that the homogenization approach to multi-droplet
patterns in a related context was first discussed in
\cite{m1:pre97}. Also, let us mention that in a related class of
problems existence of limiting density for the ground states of
particle systems interacting via potentials like our $G$ as the number
of particles goes to infinity was proved in \cite{kiessling99}. The
difference with our result, however, is that in \cite{kiessling99} the
limit is taken at fixed positive temperature, while in our case the
system's temperature (in the usual thermodynamic sense) is strictly
zero. Yet, as was pointed out in \cite{kiessling99}, the ``effective''
temperature of the system considered actually goes to zero as the
number of particles goes to infinity, making these results closely
related to ours.

\subsection{Fine structure of the transition point}
\label{sec:fine}

Finally, we briefly discuss the appearance of non-trivial minimizers
in the vicinity of the point $\bar \delta_m = \tfrac12 \sqrt[3]{9} \,
\kappa^2$ (this transition point was identified in
\cite{m:pre02,m:phd}). First, we note that by the results just
obtained the transition from trivial to non-trivial minimizers appears
to be quite abrupt in the limit $\ep \to 0$. In fact, in this limit
one goes immediately from no droplets to infinitely many droplets upon
crossing the point $\bar \delta = \bar \delta_m$ from below. 

Observe that the energy $\bar E[u]$ (and, equivalently, $E[u] -
E[-1]$) is a monotonically decreasing function of
$\bar\delta$. Therefore, a passage through the neighborhood of
$\bar\delta = \bar \delta_m$ at small but finite $\ep$ will result in a
monotonic increase of the number of droplets in a minimizer. This
number will quickly get large as one moves away from the transition
point. Therefore, in order to analyze droplet creation at the
transition, we need to further zoom in on the parameter region around
$\bar \delta = \bar \delta_m$. Let us introduce the renormalized
distance to the transition (with the transition point shifted
appropriately):
\begin{eqnarray}
  \label{eq:tau}
  \tau = {|\ln \ep| \over \kappa^2} \left( \bar \delta - {\sqrt[3]{9}
      \, \over 2} \, \kappa^2 - {\ln |\ln \ep| \over 2 \sqrt[3]{3} \,
      |\ln  \ep|} \, \kappa^2 \right),  
\end{eqnarray}
and consider the behavior of energy $\bar E$ in the limit $\ep \to 0$
for $\tau = O(1)$. As can be easily seen, all the estimates obtained
previously remain valid in this case, and minimizers are close to a
collection of $N$ disks separated by large distances, whose energy is
given by $\bar E_N$ to $O(\ep^\alpha)$. We can also write the energy
$\bar E_N$ in the form
\begin{eqnarray}
  \label{eq:Etau}
  \bar E_N = \sum_{i=1}^N \bar E_1(\bar r_i) + 4 \pi^2 |\ln
  \ep|^{-2} \sum_{i=1}^{N-1} \sum_{j=i+1}^N G(x_i - x_j) \bar r_i^2
  \bar r_j^2,
\end{eqnarray}
where 
\begin{eqnarray}
  \label{eq:108}
   |\ln \ep| \bar E_1(\bar r) = 2 \pi \left( \bar r - \tfrac12
     \sqrt[3]{9} \, \bar r^2 + \tfrac16 \bar r^4 \right) \nonumber \\ 
  + \tfrac13 \pi |\ln \ep|^{-1} \ln |\ln \ep| \, (\bar r^2 -
  \sqrt[3]{9}) \bar r^2 \nonumber \\ 
  - \pi |\ln \ep|^{-1} (\bar r^2 (\ln \bar \kappa
  \bar r - \tfrac14) + 2 \tau ) \bar r^2
\end{eqnarray}
is the energy of one disk-shaped droplet of radius $\bar r$.

It is easy to see that in the limit $\ep \to 0$ we have $\bar E_1(\bar
r) \geq 0$ for all $\bar r> 0$, and $\bar E_1(\bar r) = 0$ if and only
of $\bar r = \sqrt[3]{3}$. Therefore, $\bar r_i \to \sqrt[3]{3}$
uniformly in a minimizer as $\ep \to 0$ with $\tau$ fixed. In fact, by
convexity of $\bar E_1$ near $\bar r = \sqrt[3]{3}$ we have $\bar r_i
- \sqrt[3]{3} = O(|\ln \ep|^{-1} \ln |\ln \ep|)$ in the limit $\ep \to
0$. Therefore, we obtain (the summation is absent in the formula, if
$N = 1$)
\begin{eqnarray}
  \label{eq:EN33}
  \bar E_N = 12 \sqrt[3]{3} \, \pi^2 |\ln \ep|^{-2} \Biggl\{
  \sum_{i=1}^{N-1}  \sum_{j=i+1}^N G(x_i - x_j) \hspace{2.5cm} \nonumber
  \\  
  - {N \over 4 \pi} \Biggl( \ln \bar \kappa + \frac13 \ln 3 - \frac14 
  + {2 \tau \over \sqrt[3]{9} } \Biggr) \Biggr\} + o(|\ln \ep|^{-2}).   
\end{eqnarray}
From this expression it is easy to see that $N = O(1)$ quantity. Thus,
in this case the problem reduces to minimizing the pair interaction
potential given by the sum in (\ref{eq:EN33}). We summarize the above
discussion by stating the following result.

\begin{proposition}
  \label{p:E33}
  Let $\bar \delta$ be given by (\ref{eq:tau}) with $\tau$
  fixed. Then, there exists a strictly monotonically increasing
  sequence of numbers $(\tau_n)$, with $\tau_n \to \infty$ as $n \to
  \infty$, such that, provided that $\ep \ll 1$:
  \begin{itemize}
  \item[(i)] If $\tau < \tau_1 = \frac{1}{24} \sqrt[3]{9} \, (3 - 4
    \ln 3 - 12 \ln \bar \kappa)$, then there are no non-trivial
    minimizers of $E$.
  \item[(ii)] If $\tau_1 < \tau < \tau_2$, with $\tau_2 = \tau_1 + 2
    \pi \sqrt[3]{9} \, \min G$, the minimizer of $E$ is a single
    droplet.
  \item[(iii)] If $\tau_n < \tau < \tau_{n+1}$, all minimizers of $E$
    consist of precisely $n$ droplets. The droplet centers $\{x_i\}$
    nearly minimize $V = \displaystyle \sum_{i=1}^{n-1} \sum_{j=i+1}^n
    G(x_i - x_j)$.
  \end{itemize}
\end{proposition}


Let us mention that local minimizers of $E$ without screening
(i.e. with $\kappa \to 0$) which are close to disks of the same radius
centered at the minimizers of $V$ were constructed perturbatively in a
recent work of Ren and Wei \cite{ren07jns,ren07rmp}. We note that when
$\tau = O(1)$, existence of these solutions easily follows from our
analysis, if one notices that in the considered regime the excess
energy of a minimizing sequence controls the isoperimetric deficit of
each droplet and enforces $O(1)$ distance between them. Therefore,
solutions with a prescribed number of droplets may be obtained by
minimizing over all $u \in BV(\Omega; \{-1,1 \})$, such that the
support of $1 + u$ has a fixed number of disjoint components. In turn,
by Proposition \ref{p:E33} the global minimizers of $E$ belong to this
class.

\section{Connection to the diffuse interface energy}
\label{sec:conn}

We now turn to the study of the relationship between the sharp
interface energy $E$ and the diffuse interface energy $\mathcal
E$. Since most of our analysis here does not rely on any particular
assumptions on the dimensionality of space, we will treat the general
case of $\Omega$ being a $d$-dimensional torus: $\Omega = [0, 1)^d$.
We assume that $W$ is a symmetric double-well potential with
non-degenerate minima at $u = \pm 1$, together with some natural
technical assumptions:

\begin{enumerate}
\item[(i)] $W \in C^3(\mathbb R)$, $W(u) = W(-u)$, and $W \geq 0$, 

\item [(ii)] $W(+1) = W(-1) = 0$ and $W''(+1) = W''(-1) > 0$.

\item [(iii)] $W''(|u|)$ is monotonically increasing for $|u| \geq 1$,
  $\lim_{|u| \to \infty} W''(u) = +\infty$, and $|W'(u)| \leq C (1 +
  |u|^q)$, for some $C > 0$ and $q > 1$, with $q < {d + 2 \over d -
    2}$ if $d > 2$.
\end{enumerate}
Since we are setting the surface tension to $\ep$, we need to
additionally normalize $W$ as follows:
\begin{enumerate}
\item[(iv)] We have
\begin{eqnarray}
  \label{eq:W0}
  \int_{-1}^1 \sqrt{2 W(u)} \, du = 1.
\end{eqnarray}

\end{enumerate}
Note that these assumptions are satisfied for, e.g., the rescaled
version of the classical Ginzburg-Landau energy: $W(u) = \tfrac{9}{32}
(1 - u^2)^2$ for $d \leq 3$. Also note that this assumption is not
restrictive, since it is always possible to make (\ref{eq:W0}) hold by
an appropriate rescaling.

Let us begin our analysis with a few general observations. First of
all, it is clear from standard arguments (see e.g. \cite{struwe}) that
for any $\ep > 0$ there exists a minimizer $u \in H^1(\Omega)$ of
$\mathcal E$ satisfying $\int_\Omega u \, dx = \bar u$. Note that any
critical point $u$ of $\mathcal E$, including minimizers, is a weak
solution of the Euler-Lagrange equation (here and below $G_0$ solves
(\ref{eq:G0}) with periodic boundary conditions and has zero mean)
\begin{eqnarray}
  \label{eq:ELE}
  \ep^2 \Delta u - W'(u) - v + \mu = 0, \qquad v = \int_\Omega G_0(x - y)
  (u(y) - \bar u) dy,  
\end{eqnarray}
where
\begin{eqnarray}
  \label{eq:mu}
  \mu = \int_\Omega W'(u) \, dx
\end{eqnarray}
is the Lagrange multiplier. Furthermore, from the Sobolev imbedding
theorem we have $u \in L^p(\Omega)$ for $p = {2 d \over d - 2}$, and
hence $v \in W^{2,p}(\Omega) \subset C^{0,\alpha}(\Omega)$, for some
$\alpha \in (0, 1)$, if $d < 6$. Applying Moser iteration
\cite{struwe}, we then find that $u \in L^p(\Omega)$, for any $p <
\infty$. Therefore, by standard elliptic regularity theory
\cite{gilbarg}, we also have $u \in W^{2,p}(\Omega)$, so $u \in
L^\infty(\Omega)$ and is, in fact, a classical solution of
(\ref{eq:ELE}).

We now show that $u$ is uniformly bounded and that $|u|$ cannot much
exceed 1 whenever $\mathcal E [u]$ is sufficiently small, at least for
$d$ not too high.
\begin{proposition}
  \label{p:1del}
  Let $d < 6$ and let $u$ be a critical point of $\mathcal E$. Then,
  for every $\delta > 0$ we have $|u| < 1 + \delta$ and $|v| < \delta$
  in $\Omega$, whenever $\mathcal E[u]$ is sufficiently small.
\end{proposition}
\begin{proof}
  Observe first that for every $\delta > 0$ and $\mathcal E[u]$ small
  enough we have $|\{|u| > 1 + \delta \}| < \tfrac12$. Now, suppose
  that the maximum value $u_m$ of $|u|$ is greater than $1 +
  \delta$. Without a loss of generality, we may assume that $u_m =
  \max u$. By the preceding observation, we have $\mu \leq C +
  \tfrac12 W'(u_m)$, for some $C > 0$ independent of $u_m$. Therefore,
  in view of the monotonic increase to infinity of $W'(u)$ due to
  hypothesis (iii) on $W$, we have $\mu \leq \tfrac34 W'(u_m)$ for
  $u_m$ sufficiently large. Now, taking into account that $v \geq - C'
  u_m$ for some $C' > 0$ and large enough $u_m$, from (\ref{eq:ELE})
  we find that $\ep^2 \Delta u \geq \tfrac14 W'(u_m) - C' u_m> 0$ at
  the point where $u = u_m$, in view of assumption (iii) on $W$,
  contradicting the maximality of $u$. Finally, to see that $|u| < 1 +
  \delta$ with any $\delta > 0$, when $\mathcal E[u] \ll 1$, note that
  $u \to \pm 1$ a.e. when $\mathcal E[u] \to 0$. Hence, in view of the
  uniform bound on $u$ obtained earlier, we have $\mu \to
  0$. Furthermore, since the non-local term in the energy can be
  written as $\tfrac12 \int_\Omega |\nabla v|^2 dx$, and $v$ is
  uniformly bounded in $W^{2,p}(\Omega)$ for any $p < \infty$, we also
  have $v \to 0$ uniformly in $\Omega$ in this limit. Therefore, in
  view of positivity of $W'(1 + \delta)$, we can apply the same
  argument as above to complete the proof of the statement. 
\end{proof}

We note that while the arguments above hold for every critical point
$\mathcal E$ with small energy, it is generally possible for a local
minimizer of $\mathcal E$ to strongly deviate from $\pm 1$ in most of
$\Omega$: take, for instance, one-dimensional periodic solutions of
(\ref{eq:ELE}) with period $O(1)$
\cite{ko:jetp80,ko:book,mimura80}. Of course, these critical points
will have $O(1)$ energy when $\ep \to 0$, as opposed to minimizers of
$\mathcal E$ whose energy vanishes in this limit. Let us also mention
that numerical evidence shows that generally $\max |u| > 1$, even for
minimizers and $\ep \ll 1$. 



We now turn to estimating the minimal energy of $\mathcal E$ from
below by the minimal energy of $E$. For $u \in H^1(\Omega)$ with $\int
u \, dx = \bar u$, let us separate the domain $\Omega$ into three
pairwise-disjoint subdomains:
\begin{eqnarray}
  \label{eq:Omdel}
  \Omega = \Omega^\delta_+ \cup \Omega^\delta_- \cup \Omega^\delta_0.   
\end{eqnarray}
where
\begin{eqnarray}
  \label{eq:Ommdel}
  && \Omega^\delta_+ =  \{ x \in \Omega : u(x) \geq 1 - \delta \}, \\
  && \Omega^\delta_- = \{ x \in \Omega : u(x) \leq -1 + \delta \}, \\
  && \Omega^\delta_0 = \{ x \in \Omega : -1 + \delta < u(x) < 1 -
  \delta \}, 
\end{eqnarray}
Next, let us introduce the following three auxiliary functionals (for
simplicity of notation, we will suppress the index $\delta$ in the
definition of each functional):
\begin{eqnarray}
  \label{eq:EE1}
  \mathcal E_1[u] & = & \int_{\Omega^\delta_0} \left( \frac{\ep^2}{2}
    |\nabla u|^2 + W(u) \right) dx, \\
  \label{eq:EE2}
  \mathcal E_2[u] & = & {1 \over 2 \kappa^2} \int_{\Omega^\delta_+ \cup
    \Omega^\delta_-} (u - u_0)^2 dx \nonumber \\ 
  & + & \frac12 \int_\Omega  \int_\Omega
  (u(x) - \bar u) G_0(x - y) (u(y) - \bar u) dx dy,
\end{eqnarray}
where $u_0(x) = \pm 1$ whenever $x \in \Omega^\delta_\pm$,
respectively, with
\begin{eqnarray}
  \label{eq:kappa}
  \kappa = {1 \over \sqrt{W''(1)}},
\end{eqnarray}
and
\begin{eqnarray}
  \label{eq:EE3}
  \mathcal E_3[u] = \int_{\Omega^\delta_+ \cup
    \Omega^\delta_-} \left( W(u) - {1 \over 2 \kappa^2} (u - u_0)^2
  \right) dx. 
\end{eqnarray}
It is clear that the energy $\mathcal E$ can be estimated from below
as
\begin{eqnarray}
  \label{eq:EE}
  \mathcal E \geq \mathcal E_1 + \mathcal E_2 + \mathcal E_3.
\end{eqnarray}
Hence, we are going to establish a lower bound for $\mathcal E$ by
considering the lower bounds for each term in the sum above.

We start with the part of energy that is associated with the
interfaces:
\begin{lemma}
  \label{l:E1}
  Let $\delta > 0$ be sufficiently small, let $u \in H^1(\Omega)$ and
  suppose that $|\Omega^\delta_- \cup \Omega^\delta_+| > 0$.  Then
  there exists $u_0 \in BV(\Omega; \{-1, 1\})$ such that $u_0(x) = \pm
  1$ whenever $x \in \Omega^\delta_\pm$, and
  \begin{eqnarray}
    \label{eq:EE10}
    \mathcal E_1[u] \geq  {\ep \over 2} (1 - a_1 \delta^2) \int_\Omega
    |\nabla u_0| dx
  \end{eqnarray}
  for some $a_1 > 0$ independent of $\delta$ and $\ep$.
\end{lemma}

\begin{proof}
  First of all, if either $|\Omega^\delta_+|$ or $|\Omega^\delta_-|$
  is zero, we can simply choose $u_0$ to be constant (e.g. $u_0 = -1$
  when $|\Omega^\delta_+| = 0$). So, let us assume that both
  $|\Omega^\delta_\pm| > 0$, and approximate $u$ in $H^1(\Omega)$ by a
  piecewise linear function $\tilde u$ with $\nabla \tilde u \not = 0$
  almost everywhere in $\Omega$. Then, using the Modica-Mortola trick
  \cite{modica77,modica87} and the co-area formula \cite{attouch}, we
  find
  \begin{eqnarray}
    \label{eq:coarea}
    \mathcal E_1[\tilde u] \geq \ep \int_{|\tilde u| < 1 - \delta} \sqrt{2 
      W(\tilde u)} \, |\nabla \tilde u| \, dx = \ep \int_{-1+
      \delta}^{1-\delta} \sqrt{2 W(t)} \, |\{ \tilde u = t\}| \, dt.
  \end{eqnarray}
  Since the function $|\{ \tilde u = t \}|$ is continuous for all $t
  \in [-1 + \delta, 1 - \delta]$, there exists a constant $c \in
  [-1 + \delta, 1 - \delta]$ such that the right-hand side of
  (\ref{eq:coarea}) equals $|\{ \tilde u = c \}| \int_{-1 + \delta}^{1
    - \delta} \sqrt{2 W(t)} \, dt \geq (1 - a_1 \delta^2) |\{ \tilde u
  = c \}|$, for some $a_1 > 0$ and all $\delta$ small enough.

  Now, define $\tilde u_0 \in BV(\Omega; \{-1, 1\})$ as
  \begin{eqnarray} 
    \label{u0c}
    \tilde u_0(x) = 
    \begin{cases} 
      +1, & \tilde u(x) > c, \\
      -1, & \tilde u(x) \leq c.
    \end{cases}
  \end{eqnarray}
  The preceding arguments imply the desired inequality for $\tilde u$.
  Passing to the limit in the approximation, we obtain the result,
  with $u_0 = \lim \tilde u_0$ in $L^1(\Omega)$ upon extraction of a
  subsequence.  
\end{proof}

\begin{lemma}
  \label{l:E2}
  Let $u$ and $u_0$ be as in Lemma \ref{l:E1}, let $u$ satisfy
  $\int_\Omega u \, dx = \bar u$, and let $|u| \leq 1 + \delta^3$ and
  $\left| \int_\Omega G_0(x - y) (u(y) - \bar u) \, dy \right| \leq
  \delta^3$ in $\Omega$, for $\delta > 0$ sufficiently small.  Then
  \begin{eqnarray}
    \label{eq:EE0}
    \mathcal E_2[u] \geq \frac12 \int_\Omega
    \int_\Omega (u_0(x) - \bar u) G(x - y) (u_0(y) - \bar u) \, dx dy
    - a_2 \delta \,  \mathcal E[u], 
  \end{eqnarray}
  for some $a_2 \geq 0$ independent of $\delta$ and $\ep$, whenever
  $\mathcal E[u] \leq \delta^{\frac32 (d + 6)}$.
\end{lemma}

\begin{proof}
  Let us write $u$ as follows
  \begin{eqnarray}
    \label{eq:u012}
    u = u_0 + u_1 + u_2, \qquad u_1(x) = -\kappa^2 \int_\Omega G_0(x -
    y) (u(y) - \bar u) dy.
  \end{eqnarray}
  Note that by assumption, $||u_1||_{L^\infty(\Omega)} \leq C
  \delta^3$ and $||u_2||_{L^\infty(\Omega)} \leq C$, for some $C > 0$.
  Now, observe that $u_1$ solves
  \begin{eqnarray}
    \label{eq:36}
    -\Delta u_1 + \kappa^2 u_1 = -\kappa^2 (u_0 + u_2 - \bar u), 
  \end{eqnarray}
  and, therefore, we also have
  \begin{eqnarray}
    \label{eq:37}
    u_1(x) = -\kappa^2 \int_\Omega G(x - y) (u_0(y) + u_2(y) - \bar u)
    dy. 
  \end{eqnarray}
  Substituting $u$ in the form (\ref{eq:u012}) into (\ref{eq:EE2}), we
  obtain
  \begin{eqnarray}
    \label{eq:EEE0}
    \mathcal E_2(u_0 + u_1 + u_2) & = &   {1 \over 2 \kappa^2}
    \int_{\Omega^\delta_+ \cup 
      \Omega^\delta_-} (u_1 + u_2)^2 \, dx \nonumber \\ 
    && \hspace{-3cm}
    + \frac12 \int_\Omega \int_\Omega (u_0(x) + u_1(x) + u_2(x) - \bar 
    u) G(x - y) (u_0(y) + u_2(y) - \bar u) dy dx \nonumber \\ 
    && \hspace{-3cm} = -{1 \over 2 \kappa^2} \int_{\Omega^\delta_0}
    u_1^2 \, dx - {1 \over \kappa^2} \int_{\Omega^\delta_0} u_1 u_2 \,
    dx - \frac12 \int_\Omega \int_\Omega u_2(x) G(x - y) u_2(y) dy dx
    \nonumber \\ 
    && \hspace{-3cm} + \frac12 \int_\Omega \int_\Omega (u_0(x) 
    - \bar u) G(x - y)  (u_0(y) - \bar u) dy dx  + {1 \over 2
      \kappa^2} \int_{\Omega^\delta_+ \cup \Omega^\delta_-} 
    u_2^2 \, dx \nonumber \\ 
    && \hspace{-3cm} \geq \frac12 \int_\Omega \int_\Omega (u_0(x) 
    - \bar u) G(x - y)  (u_0(y) - \bar u) dy dx -{1 \over 2 \kappa^2}
    \int_{\Omega^\delta_0} u_1^2 \, dx  \nonumber \\
    && \hspace{-3cm} + \int_{\Omega^\delta_0} \int_\Omega u_2(x) G(x -
    y) (u_0(y) - \bar u) \, dy dx \nonumber \\
    && \hspace{-3cm} + \frac{1}{2 \kappa^2} \int_{\Omega^\delta_+ \cup
      \Omega^\delta_-}\left( u_2^2(x) - \kappa^2  u_2(x)
      \int_{\Omega^\delta_+ \cup \Omega^\delta_-} 
      G(x - y) u_2(y) dy 
    \right) dx.
  \end{eqnarray}
  In fact, the last line in (\ref{eq:EEE0}) is non-negative. Indeed,
  writing the integral in the last line of (\ref{eq:EEE0}) with the
  help of the Fourier Transform $\hat a_q$ of $\tilde u = u_2
  \chi_{\Omega^\delta_+ \cup \Omega^\delta_-}$, where
  $\chi_{\Omega^\delta_+ \cup \Omega^\delta_-}$ is the characteristic
  function of $\Omega^\delta_+ \cup \Omega^\delta_-$:
  \begin{eqnarray}
    \label{eq:u2q}
    \hat a_q = \int_{\Omega^\delta_+ \cup \Omega^\delta_-} e^{i q
      \cdot x} u_2(x) \, dx, 
  \end{eqnarray}
  we obtain
  \begin{eqnarray}
    \label{eq:EE0l}
    \int_{\Omega^\delta_+ \cup \Omega^\delta_-} 
    \left( u_2^2(x) - \kappa^2  u_2(x) \int_{\Omega^\delta_+ \cup
        \Omega^\delta_-} G(x - y) u_2(y) dy 
    \right) dx \nonumber \\ 
    = \sum_{q \in 2 \pi \mathbb Z^d} {|q|^2 |\hat a_q|^2 \over \kappa^2
      + |q|^2}  \geq 0.  
  \end{eqnarray}

  To estimate the remaining terms in (\ref{eq:EEE0}), we note that 
  \begin{eqnarray}
    \label{eq:42}
    \left| {1 \over \kappa^2} \int_{\Omega^\delta_0} u_1 u_2 \, dx +
    \int_{\Omega^\delta_0} \int_\Omega u_2(x) G(x - y) (u_0(y) - \bar
    u) \, dy dx \right|  \nonumber \\ 
    =   \left| \int_{\Omega^\delta_0} \int_\Omega u_2(x) G(x -
    y) u_2(y) \, dy dx \right| \nonumber \\ 
    \leq \int_{\Omega^\delta_0}  \int_{\{|u| \geq 1 - \delta^3 \}}
    |u_2(x)| G(x - y) |u_2(y)| \, dy dx \nonumber \\   
    + \int_{\Omega^\delta_0}
    \int_{\{ |u| < 1 - \delta^3 \} }  |u_2(x)| G(x - y) |u_2(y)| \,
    dy dx.  
  \end{eqnarray}
  Since $G \in L^p(\Omega)$ for all $p < \frac{d}{d-2}$ (any $p <
  \infty$ in $d = 2$), by H\"older inequality we can see that for
  any $\tilde \Omega \subseteq \Omega$
  \begin{eqnarray}
    \label{eq:50}
    \int_{\tilde \Omega} G(x - y) |u_2(y)| \, dy \leq C
    \left( \int_{\tilde \Omega} |u_2|^q \, dx \right)^{1/q} \leq C
    ||u_2||_{L^\infty({\tilde \Omega})} |\tilde \Omega|^{1/q}, 
  \end{eqnarray}
  for any $q > \frac{d}{2}$. Therefore, continuing the estimates in
  (\ref{eq:42}), we obtain
  \begin{eqnarray}
    \label{eq:102}
    \left| \int_{\Omega^\delta_0} \int_\Omega u_2(x) G(x -
      y) u_2(y) \, dy dx \right| \nonumber \\
    \leq C \left( \delta^3 + \delta^{-6/q} \mathcal E^{1/q}[u]
    \right) |\Omega_0^\delta| \leq 2 C \delta^3 |\Omega_0^\delta|,
  \end{eqnarray}
  whenever $\mathcal E[u] \leq \delta^{3 (q + 2)}$, where we took into
  account that by the assumptions of the lemma $|u_2| \leq |u - u_0| +
  |u_1| \leq C \delta^3$ in $\{ |u| > 1 - \delta^3 \}$, and that
  $\mathcal E[u] \geq \int_{\{ |u| < 1 - \delta^3 \} } W(u) \, dx \geq
  c \delta^6 |\{ |u| < 1 - \delta^3 \} |$ for some $c > 0$.

  Similarly, we have
  \begin{eqnarray}
    \label{eq:9}
    \left| \int_{\Omega^\delta_0} u_1^2 \, dx \right| + \left|
      \int_{\Omega^\delta_0} u_1 u_2 \, dx \right| \leq 
    C \delta^3 |\Omega^\delta_0|.
  \end{eqnarray}
  The statement of the lemma then follows from the fact that $\mathcal
  E[u] \geq \int_{\Omega^\delta_0} W(u) dx \geq c \delta^2
  |\Omega^\delta_0|$, for some $c > 0$, by choosing $q = \tfrac{1}{2}
  (d + 2)$. 
\end{proof}

\begin{lemma}
  \label{l:E3}
  Let $u$ and $u_0$ be as in Lemma \ref{l:E1}. Then
  \begin{eqnarray}
    \label{eq:E3}
    \mathcal E_3[u] \geq - a_3 \delta \, \mathcal E[u].
  \end{eqnarray}
  for some $a_3 \geq 0$ independent of $\delta$ and $\ep$, for
  sufficiently small $\delta > 0$.
\end{lemma}

\begin{proof}
  By assumption (iii) on $W$, we have $W(u) \geq \frac{1}{2 \kappa^2}
  (u - u_0)^2$ whenever $|u| > 1$. Hence
  \begin{eqnarray}
    \label{eq:18}
    \mathcal E_3[u] \geq \int_{\{1 - \delta \leq |u| \leq 1 \}} \left(
      W(u) - \frac{1}{2 \kappa^2} (u - u_0)^2 \right) dx \nonumber \\ 
    \geq - C \delta 
    \int_{\{1 - \delta \leq |u| \leq 1 \}} (u - u_0)^2 dx \geq  - a_3
    \delta \, \mathcal E[u],     
  \end{eqnarray}
  for some $a_3 \geq 0$.
\end{proof}

Combining all the results above with an observation that
$|\Omega^\delta_+ \cup \Omega^\delta_-| > 0$ for $\mathcal E[u]$ small
enough, we obtain
\begin{proposition}
  \label{p:lower}
  Let $\delta > 0$ be sufficiently small, let $u \in H^1(\Omega)$
  satisfy $\int_\Omega u \, dx = \bar u$, let $|u| \leq 1 + \delta^3$
  and $\left| \int_\Omega G_0(x - y) (u(y) - \bar u) \, dy \right|
  \leq \delta^3$ in $\Omega$, and let $\mathcal E \leq \delta^{\frac32
    (d + 6)}$. Then there exists a function $u_0 \in BV(\Omega; \{-1,
  1\})$ such that $\mathcal E[u] \geq (1 - \delta^{1/2}) E[u_0]$, with
  $\kappa$ given by (\ref{eq:kappa}).
\end{proposition}
Importantly, the lower bound in Proposition \ref{p:lower} is sharp in
the limit $\ep \to 0$ for all functions $u_0 \in BV(\Omega; \{-1,1\})$
obeying suitable bounds (satisfied by minimizers of $E$ in $d = 2$):

\begin{proposition}
  \label{p:upper}
  Let $u_0 \in BV(\Omega; \{-1, 1\})$, with the jump set of class
  $C^2$, let the principal curvatures of the jump set of $u_0$ be
  bounded by $\ep^{-\alpha}$ for some $\alpha \in [0, 1)$, let the
  distance between different connected portions of the jump set be
  bounded by $\ep^\alpha$, and let $\left| \int_\Omega G(x - y)
    (u_0(y) - \bar u) \, dy \right| \leq \delta$ for some $\delta > 0$
  small enough. Then there exists a function $u \in H^1(\Omega)$ with
  $\int_\Omega u \, dx = \bar u$, such that $\mathcal E[u] \leq (1 +
  \delta^{1/2}) E[u_0]$, with $\kappa$ given by (\ref{eq:kappa}),
  whenever $E[u_0] \leq \delta^{\frac12 (d + 3)}$ and $\ep \ll 1$.
\end{proposition}
\begin{proof} 
  For simplicity of presentation, we only give the proof in the case
  $d = 2$. With minor modifications, the proof remains valid for all
  $d$.

  Here we adapt the standard construction of a trial function for the
  local part of the Ginzburg-Landau energy.  Let $U(\rho)$ be the
  solution of the ordinary differential equation
  \begin{eqnarray}
    \label{eq:front}
    {d^2 U \over d \rho^2} - W'(U) = 0, \qquad U(-\infty) = 1, \quad
    U(+\infty) = -1, \quad U(0) = 0,
  \end{eqnarray}
  where the last condition fixes translations. As is well-known (see
  e.g. \cite{fife77}), this solution exists, is unique and is a
  strictly monotonically decreasing odd function, approaching the
  equilibria at $\rho = \pm \infty$ exponentially fast. Therefore, for
  any $\delta > 0$ we have $|U(\rho)| \leq 1 - \delta$, if and only if
  $|\rho| \leq l$, with some positive $l = O(|\ln \delta|)$. Also note
  that by hypothesis (iv) on $W$
  \begin{eqnarray}
    \label{eq:100}
    \int_{-l}^l \left\{ \frac12 \left| {d U \over d \rho} \right|^2 +
      W(U) \right\} d\rho = \int_{-1+\delta}^{1-\delta} \sqrt{2 W(s)}
    \, ds = 1 + O(\delta^2). 
  \end{eqnarray}

  Now, introduce the signed distance function $r(x) = \pm
  \mathrm{dist} (x, \Omega^\pm)$, where $\Omega^\pm = \{ u_0 = \pm 1
  \}$, whenever $x \in \Omega^\mp$, and define a regularized version
  $u_0^\ep$ of $u_0$:
  \begin{eqnarray}
    \label{eq:103}
    u_0^\ep(x) = 
    \begin{cases}
      U(\ep^{-1} r(x)), & |r(x)| \leq \ep l, \\
      (1 - \delta + \ep^{-1} \delta (|r(x)| - \ep l) ) u_0(x), & \ep l
      \leq |r(x)|  \leq \ep (l + 1), \\
      \, u_0(x), & |r(x)| \geq \ep (l + 1).
    \end{cases}
  \end{eqnarray}
  Then, it is easy to see that the function
  \begin{eqnarray}
    \label{eq:uuhat}
    u(x) = u_0^\ep(x) - \kappa^2 \int_\Omega G(x - y) (u_0^\ep(y) -
    \bar u) dy.  
  \end{eqnarray}
  is in $H^1(\Omega)$, with $\int_\Omega u \, dx = \bar u$. Moreover,
  we have for any $q > 1$
  \begin{eqnarray}
    \label{eq:39}
    |u(x) - u_0^\ep(x)| \leq \kappa^2 \left| \int_\Omega G(x - y)
      (u_0(y) - \bar u) \, dy \right| \nonumber \\
    + \kappa^2 \int_{\Omega_l} G(x - y) |u_0^\ep(y) - u_0(y)| \, dy 
    \nonumber \\ 
    \leq C (\delta + |\Omega_l|^{1/q}) \leq C' (\delta +
    |\ln \delta|^{1/q} E^{1/q}[u_0]) \leq C'' \delta,
  \end{eqnarray}
  where we defined $\Omega_l = \{ |r| \leq \ep (l + 1) \}$, estimated
  $\int_{\Omega_l} G(x - y) (u_0^\ep(y) - u_0(y)) \, dy$ as in Lemma
  \ref{l:E2}, and used the curvature bound and the assumption on $E$
  with $d = 2$ and $q = 2$.
  On the other hand, after a few integrations by parts, from
  (\ref{eq:uuhat}) we also obtain
  \begin{eqnarray}
    \label{eq:49}
    \int_\Omega |\nabla (u - u_0^\ep)|^2 \, dx = - \kappa^2
    \int_\Omega (u - u_0^\ep) ( u - \bar u) \, dx \hspace{2cm}
    \nonumber \\ 
    = \kappa^4 \int_\Omega (u(x) - \bar u) G_0(x - y) (u(y) - \bar u) 
    \, dy dx \leq 2 \kappa^4 \mathcal E[u].
  \end{eqnarray}

  To estimate $\mathcal E[u]$, let us introduce a system of curvilinear
  coordinates $(\rho, \xi)$ consisting of the signed distance $\rho$
  to the jump set of $u_0$ and the projection $\xi$ onto the jump
  set. By our assumptions this is possible whenever $|r(x)| <
  \ep^\alpha$. Therefore, for $\ep \ll 1$ we can write
  \begin{eqnarray}
    \label{eq:EEup}
    \mathcal E[u] = \int_{-\ep (l+1)}^{\ep(l+1)} \int_{\partial
      \Omega^+} \Biggl( \frac{\ep^2}{2} \left| {\partial u
        \over \partial \rho} \right|^2 + \frac{\ep^2}{2} (1 + \rho
    K )^{-2} \left| {\partial u \over \partial \xi} \right|^2 
    + W(u) \Biggr)       \nonumber \\ 
    \times \left(1 + \rho K \right)  \, d \mathcal H^1(\xi) \,
    d \rho + \int_{\Omega \backslash \Omega_l } \Biggl( {\ep^2 \over
      2} 
    |\nabla u|^2 + W(u) \Biggr) \, dx \nonumber \\
    + \frac12 \int_\Omega \int_\Omega (u(x) - \bar u) G_0(x - y) (u(y) -
    \bar u) \, dy dx,
  \end{eqnarray}
  where $K = K(\xi)$ is the curvature at point $\xi$ on the jump set
  of $u_0$. Substituting the ansatz of (\ref{eq:uuhat}) into
  (\ref{eq:EEup}), taking into account that $|\nabla (u - u_0^\ep) |
  \leq C$ in $\Omega_l$ for some $C > 0$ independent of $\ep$ (we have
  $u - u_0^\ep$ uniformly bounded in $W^{2,p}(\Omega)$, for any $p <
  \infty$) and that $\nabla u = \nabla (u - u_0) = \nabla (u -
  u_0^\ep)$ in $\Omega \backslash \Omega_l$, and using (\ref{eq:39})
  and (\ref{eq:49}), we obtain (estimating each line in
  (\ref{eq:EEup}) separately)
  \begin{eqnarray}
    \label{eq:105}
    \mathcal E[u] = \frac{\ep}{2} \bigl(1 + O(\ep^{1-\alpha} |\ln
    \delta|) + O(\delta |\ln \delta|) 
    \bigr) \int_\Omega |\nabla u_0| \, dx \nonumber \\
    + \frac{1}{2 \kappa^2} \int_{ \Omega \backslash \Omega_l} (u -
    u_0)^2 \, dx + O(\ep^2 \mathcal E[u]) + O(\delta \, \mathcal E[u]) 
    \nonumber \\ 
    + \frac12 \int_\Omega \int_\Omega (u(x) - \bar u) G_0(x - y)
    (u(y) -  \bar u) \, dy dx.
  \end{eqnarray}
  Now, using the identity
  \begin{eqnarray}
    \label{eq:104}
    \int_\Omega \int_\Omega (u(x) - \bar u) G_0 (x - y) (u(y)
    - \bar u) \, dy dx + \kappa^{-2} \int_\Omega (u -
    u_0^\ep)^2 dx \nonumber \\ 
    = \int_\Omega \int_\Omega (u_0^\ep(x) - \bar u) G(x - y)
    (u_0^\ep(y) - \bar u) \, dy dx,
  \end{eqnarray}
  we can further write (\ref{eq:105}) as
  \begin{eqnarray}
    \label{eq:106}
    \mathcal E[u] = E[u_0] + O(\delta |\ln \delta| \, E[u_0]) +
    O(\delta \, \mathcal E[u]) 
    - \frac{1}{2 \kappa^2} \int_{\Omega_l} (u - u_0^\ep)^2 \, dx
    \nonumber \\  
    + \int_{\Omega_l} \int_\Omega (u_0^\ep(x) - u_0(x)) G(x - y)
    (u_0^\ep(y) - \bar u) \, dy dx \nonumber \\ 
    + \frac12 \int_{\Omega_l}
    \int_{\Omega_l} (u_0^\ep(x) - u_0(x)) G(x - y)  (u_0^\ep(y) -
    u_0(y)) \, dy dx,
 \end{eqnarray}
 for $\ep \ll 1$. Finally, using the same estimates as in
 (\ref{eq:39}), we obtain 
 \begin{eqnarray}
   \label{eq:107}
   (1 + O(\delta)) \mathcal E[u] = (1 + O(\delta |\ln \delta|)) E[u_0] 
   + O(\delta |\Omega_l|) + O(|\Omega_l|^{3/2}) \nonumber \\ 
    = (1 + O(\delta |\ln \delta|) ) E[u_0],
 \end{eqnarray}
 from which the result follows immediately.  
\end{proof}

The last two propositions show asymptotic equivalence of the diffuse
interface energy $\mathcal E$ with the sharp interface energy $E$ for
sufficiently well-behaved critical points and $\ep \ll 1$. In
particular, the energies of minimizers of both $E$ and $\mathcal E$
are asymptotically the same in the limit $\ep \to 0$. It would also be
natural to think that the minimizers (even local, with low energy) of
$\mathcal E$ are, in some sense, close to minimizers of $E$ when $\ep
\ll 1$ (this will be a subject of future study).

\section{Proof of the theorems}
\label{sec:proofs}

Here we complete the proofs of Theorems \ref{t:cC}--\ref{t:ok}.

\paragraph{Proof of Theorem \ref{t:cC}}
The main point of the proof is the lower bound in (\ref{eq:cC}), since
the upper bound is easily obtained by constructing a suitable trial
function (as in Lemma \ref{l:ub0}). The basic tool for the lower bound
is a kind of interpolation inequality obtained in Lemma
\ref{l:interp}. Note that the proof for $E$ works in any space
dimension.

To prove the lower bound, let us denote by $u$ a minimizer of
$E$. Introducing
\begin{eqnarray}
  \label{eq:uq}
  \hat a_q = \int_\Omega e^{i q \cdot x} (u(x) - \bar u)\, dx,
\end{eqnarray}
where $q \in 2 \pi \mathbb Z^d$, we can estimate the energy of the
minimizer as follows
\begin{eqnarray}
  \label{eq:44}
  \min E & \geq & \frac12 \int_\Omega \int_\Omega (u(x) - \bar u) 
  G(x - y) (u(y) - \bar u) \, dx dy \nonumber \\ 
  & = & \frac12 \sum_q {|\hat a_q|^2 \over \kappa^2 + |q|^2}
  \geq {|\hat a_0|^2 \over 2 \kappa^2} \nonumber \\
  & = & \frac{1}{2 \kappa^2} \left( \int_\Omega (u - \bar
    u) \, dx \right)^2 =  {2 \over \kappa^2} \left( |\Omega^+| - {1
      +  \bar u  \over 2} \right)^2, 
\end{eqnarray}
where we introduced the set $\Omega^+ = \{ u = +1 \}$.

In view of the upper bound in (\ref{eq:cC}), it follows from
(\ref{eq:44}) that $|\Omega^+| = \tfrac12 (1 + \bar u) +
O(\ep^{2/3})$, implying that $|\Omega^+|$ is bounded away from 0 or 1
for $\ep \ll 1$.  Hence, by isoperimetric inequality there exists $p >
0$ such that
\begin{eqnarray}
  \label{eq:40}
  P = \int_\Omega |\nabla u| \, dx \geq p,
\end{eqnarray}
whenever $\ep \ll 1$.  Applying Lemma \ref{l:interp} to $u - \bar u$,
we conclude that
\begin{eqnarray}
  \label{eq:43}
  \min E \geq \ep P + {C \over P^2}, 
\end{eqnarray}
for some $C > 0$ independent of $\ep$, for $\ep \ll 1$. The result
then follows from an application of Young inequality and Propositions
\ref{p:1del} and \ref{p:lower}.  \qed

\paragraph{Proof of Theorem \ref{t:main}}
This theorem combines a number of results proved in
Sec. \ref{sec:sharp} in the original, unscaled variables. Part (i) of
the theorem is the statement of Proposition \ref{p:db}. Part (ii) of
the theorem is the collection of results from Lemma \ref{l:ub0}
(taking into account that $E[u_k] < E[-1] = \tfrac12 \ep^{4/3} |\ln
\ep|^{2/3} \kappa^{-2} \bar\delta^2$ for $\bar \delta > \tfrac12
\sqrt[3]{9} \, \kappa^2$), Corollary \ref{c:c3}, Lemma \ref{l:N}, and
Propositions \ref{p:disk}, \ref{p:EEE}, and \ref{p:beta} with $\alpha
= \tfrac13 - \sigma$. Part (iii) of the theorem is contained in the
statements of Propositions \ref{p:ri} and \ref{p:dens}. \qed

\paragraph{Proof of Theorem \ref{t:ok}}

First of all, we have $\min \mathcal E \ll 1$ when $\ep \ll 1$ and
$\bar u = -1 + O (\ep^{2/3} |\ln \ep|^{1/3})$, since $\min \mathcal E
\leq \mathcal E(\bar u) = O(\ep^{4/3} |\ln \ep|^{2/3})$ in that
case. Then, from Proposition \ref{p:1del} and Lemma \ref{l:K} we
conclude that the assumptions of Propositions \ref{p:lower} and
\ref{p:upper} are satisfied for the minimizers of $\mathcal
E$. Therefore, the energies $E$ and $\mathcal E$ are asymptotically
the same in the considered limit, and the conclusion follows from
Theorem \ref{t:main} (the case $\bar\delta = \tfrac12 \sqrt[3]{9} \,
\kappa^2$ is included by monotone decrease of $\bar E$ with
$\bar\delta$). \qed

\section*{Acknowledgments}
\label{sec:acknowledgements}

The author would like to acknowledge valuable discussions with
M. Kiessling, H. Kn\"upfer, V. Moroz, M. Novaga and G. Orlandi. This
work was supported, in part, by NSF via grant DMS-0718027.

\appendix

\section{Upper bound}
\label{sec:ub0}

Here we construct a trial function that achieves the lower bound for
the energy of the non-trivial minimizers of $E$.

\begin{lemma}
  \label{l:ub0}
  Let $\bar u = -1 + \ep^{2/3} |\ln \ep|^{1/3} \bar \delta$, with
  $\bar\delta > \tfrac12 \sqrt[3]{9} \, \kappa^2$ fixed. Then there
  exists $u \in BV(\Omega; \{-1, 1\})$, such that
  \begin{eqnarray}
    \label{eq:En}
    E[u] =  \ep^{4/3} |\ln \ep|^{2/3} \left\{ {\sqrt[3]{9} \over 2}
      \left( \bar\delta - \frac14 \sqrt[3]{9} \, \kappa^2 \right) +
      O \left({\ln |\ln \ep| \over |\ln \ep|} \right) \right\}, 
  \end{eqnarray}
  for $\ep \ll 1$.
\end{lemma}

\begin{proof}
  First, consider $u_1(x) = -1 + 2 \chi_{B_r(0)}(x)$, where
  $\chi_{B_r(0)}$ is the characteristic function of a disk of radius
  $r$ centered at the origin. If $v_1(x) = \int_\Omega G(x - y)
  (u_1(y) - \bar u) \, dy$, then by using \eqref{eq:G} we explicitly
  have (see (\ref{eq:vB}))
  \begin{eqnarray}
    \label{eq:82}
    v_1(x) = - {1 + \bar u \over \kappa^2} + {2 \over \kappa^2 } 
     (1 - \kappa r K_1(\kappa r)
    I_0(\kappa |x|)), \hspace{3cm} \\
    + {2 \over \kappa} 
    \sum_{\mathbf n \in \mathbb Z^2 \backslash \{0\}}  r I_1(\kappa r)
    K_0(\kappa |x + \mathbf n|)), \qquad |x| \leq r,
  \end{eqnarray}
  where $K_n$ and $I_n$ are the modified Bessel functions of the first
  and second kind. Therefore, expanding the Bessel functions for $r
  \ll 1$ \cite{abramowitz}, we can write for $|x| \leq r$
  \begin{eqnarray}
    \label{eq:83}
    v_1(x) = - {1 + \bar u \over \kappa^2} - {r^2 \over 2} \left(2 \ln 
      \kappa r  + 2 \gamma - \ln 4 - 1 \right) - {|x|^2 \over 2}
    \qquad \nonumber \\  
    + r^2 \sum_{\mathbf n \in 
      \mathbb Z^2 \backslash \{0\}} K_0(\kappa |x + \mathbf n|) +
    O(r^4 |\ln r|),
 \end{eqnarray}
 where $\gamma \approx 0.5772$ is the Euler's constant. Substituting
 this expression into the definition of $E$, after integration we get
 \begin{eqnarray}
   \label{eq:84}
   E[u_1] = 2 \pi \ep r + \tfrac12 (1 + \bar u)^2 \kappa^{-2} - 2 \pi
   (1 + \bar u) \kappa^{-2} r^2 \hspace{3cm} \nonumber \\ 
   - \pi r^4 ( \ln \kappa r + \gamma - \ln 2 - \tfrac14)  + \pi r^4
   \sum_{\mathbf n \in  \mathbb Z^2 \backslash \{0\}} K_0(\kappa
   |\mathbf n|) + O(r^6 |\ln r|).
 \end{eqnarray}

 Now, consider a new test function
 \begin{eqnarray}
   \label{eq:86}
   u_k(x) = -1 + 2 \sum_{k_1 = 1}^k \sum_{k_2 =
     1}^k \chi_{B_{r}(\mathbf e_1 (k_1 - \frac12) + \mathbf e_2 (k_2
     - \frac12) )}(x),
 \end{eqnarray}
 consisting of $k^2$ disks of radius $r$ arranged periodically in
 $\Omega$ (here $\mathbf e_1$ and $\mathbf e_2$ are the unit vectors
 along the coordinate axes). We have
 \begin{eqnarray}
   \label{eq:85}
   E[u_k] =  \tfrac12 (1 + \bar u)^2
   \kappa^{-2}  + \pi k^2 \Bigl( 2 \ep r - 
   2  (1 + \bar u) \kappa^{-2} r^2 \nonumber \\  
   - r^4 ( \ln \kappa r + \gamma - \ln 2 - \tfrac14)
  + r^4 \sum_{\mathbf n \in  \mathbb Z^2 \backslash \{0\}}
   K_0(\kappa k^{-1} |\mathbf n|) \Bigl)  + O(k^2 r^6 |\ln r|). 
 \end{eqnarray}
 Approximating the sum in \eqref{eq:85} by an integral:
 \begin{eqnarray}
   \label{eq:87}
   k^{-2} \sum_{\mathbf n \in  \mathbb Z^2 \backslash \{0\}}
   K_0(\kappa k^{-1} |\mathbf n|) = \int_{\mathbb R^2} K_0(\kappa |x|)
   dx + O(k^{-2} \ln k) \nonumber \\ 
   = 2 \pi \kappa^{-2} + O(k^{-2} \ln k),
 \end{eqnarray}
 and expanding for $r \ll 1$, we can further write
 \begin{eqnarray}
   \label{eq:88}
   E[u_k] = \tfrac12 (1 + \bar u)^2
   \kappa^{-2}  + \pi k^2 \Bigl( 2 \ep r - 2 
   (1 + \bar u) \kappa^{-2} r^2 \nonumber \\  
   - r^4 \ln r  + 2 \pi \kappa^{-2} r^4 k^2 \Bigr) + O(k^2 r^4 \ln k).   
 \end{eqnarray}
 We now substitute $r = \ep^{1/3} |\ln \ep|^{-1/3} \sqrt[3]{3}$ into
 the expression above. Using also the definition in (\ref{eq:dbar}),
 we can write
 \begin{eqnarray}
   \label{eq:90}
   E[u_k] = \ep^{4/3} |\ln \ep|^{2/3} \Bigl( \tfrac12 \kappa^{-2}
   \bar\delta^2 - 2 \pi \sqrt[3]{9} |\ln \ep|^{-1} \kappa^{-2}
   \Bigl(\bar\delta - \tfrac12 \sqrt[3]{9} \, \kappa^2 \Bigr)  k^2
   \nonumber \\ 
   + 6 \pi^2 \sqrt[3]{3} \kappa^{-2} |\ln \ep|^{-2} k^4 \Bigr) +
   O(\ep^{4/3} |\ln \ep|^{-4/3} k^2 \ln k). 
 \end{eqnarray}
 Finally, setting
 \begin{eqnarray}
   \label{eq:91}
   k^2 = {|\ln \ep| \over 2 \pi \sqrt[3]{9}} \Biggl (\bar\delta -
   {\sqrt[3]{9} \over 2} \, \kappa^2 \Biggr) + O(1),
 \end{eqnarray}
 we obtain (\ref{eq:En}) with $u = u_k$.  
\end{proof}

Let us also quote without proof a similar result concerning the upper
bound for the reduced energy $E_N$.

\begin{lemma}
  \label{l:ubN}
  Let $\bar u = -1 + \ep^{2/3} |\ln \ep|^{1/3} \bar \delta$, with
  $\bar \delta > \tfrac12 \sqrt[2]{9} \, \kappa^2$ fixed. Then
  \begin{eqnarray}
    \label{eq:101}
    \min E_N \leq -\frac{1}{2 \kappa^2} \, \ep^{4/3} |\ln
    \ep|^{2/3} \left(\bar 
      \delta - \frac{\sqrt[3]{9}}{2} \, \kappa^2 \right)^2 +
    O \left({\ep^{4/3} \ln |\ln \ep| \over  |\ln \ep|^{1/3}} \right).  
  \end{eqnarray}
\end{lemma}

\section{Interpolation inequality}
\label{sec:interp}

Here we present the lemma that connects the non-local part of the
energy with the interfacial energy via a kind of an interpolation
inequality between $BV(\Omega)$, $H^{-1}(\Omega)$ and
$L^\infty(\Omega)$, for functions bounded away from zero.

\begin{lemma}
  Let $u \in BV(\Omega)$, where $\Omega = [0, 1)^d$ is a torus, and
  assume that $m \leq |u| \leq M$ in $\Omega$ for some $M \geq m >
  0$. Let also $\int_\Omega |\nabla u| \, dx \geq p > 0$, and let $G$
  solve (\ref{eq:Gk}) in $\Omega$ with periodic boundary
  conditions. Then there exists a constant $C = C(d, \kappa / p,
  m^2/M) > 0$ such that
  \begin{eqnarray}
    \label{eq:17}
    \int_\Omega \int_\Omega u(x) G(x - y) u(y) \, dx \, dy \geq C
    \left( \int_\Omega |\nabla u| \, dx \right)^{-2}.
  \end{eqnarray}
  \label{l:interp}
\end{lemma}

\begin{proof}
  First, extend $u$ periodically to the whole of $\mathbb R^d$. Then,
  introducing $\chi_\delta(x) = \delta^{-d} |B_1|^{-1}
  \chi(\delta^{-1} x)$, where $\chi$ is the characteristic function of
  the unit ball $B_1$ centered at the origin, we have
  \begin{eqnarray}
    \label{eq:38}
    \int_\Omega \int_{\mathbb{R}^n} u(x) \chi_\delta(x - y) u(y) \, dy
    \, dx = {1 \over |B_1|} \int_\Omega \int_{B_1} u(x) u(x + \delta
    y)  \, dy \, dx \nonumber \\ 
    \geq m^2 - {M \delta \over |B_1|} \int_\Omega \int_{B^1}
    \int_0^1 |\nabla u(x  +
    \delta t y)| \, dt \, dy \, dx \geq m^2 - M \delta \int_\Omega
    |\nabla u| \, dx, 
  \end{eqnarray}
  where the inequality is obtained by approximating $u$ by $C^1$
  functions and passing to the limit. Therefore, choosing
  \begin{eqnarray}
    \label{eq:del2}
    \delta = \left( {2 M \over m^2} \int_\Omega |\nabla u| \, dx
    \right)^{-1}, 
  \end{eqnarray}
  we obtain
  \begin{eqnarray}
    \label{eq:23}
    \frac{m^2}{2} \leq \int_\Omega \int_{\mathbb{R}^n} u(x) 
    \chi_\delta(x - y) u(y) \, dx \, dy  = \sum_q \hat \chi_\delta(q)
    |\hat u_q|^2,
  \end{eqnarray}
  where we introduced Fourier transform $\hat u_q$ of $u$:
  \begin{eqnarray}
    \label{eq:8}
    \hat u_q = \int_\Omega e^{i q \cdot x} u(x) \, dx,
  \end{eqnarray}
  with $q \in 2 \pi \mathbb Z^d$. The Fourier transform $\hat
  \chi_\delta$ of $\chi_\delta$ is, in turn, explicitly given by
  \begin{eqnarray}
    \label{eq:31}
    \hat \chi_\delta(q) = \left( {2 \over \delta |q|} \right)^{d/2}
    \Gamma \left( {d \over 2} + 1 \right) J_{d/2} (\delta |q|) ,
  \end{eqnarray}
  where $J_{d/2}(x)$ is the Bessel function of the first kind and
  $\Gamma(x)$ is the gamma-function. Now, applying Cauchy-Schwarz
  inequality, we obtain
  \begin{eqnarray}
    \label{eq:33}
    \frac{m^4}{4} \leq \left( \sum_q {|\hat u_q|^2 \over \kappa^2 +
        |q|^2} \right) \left( \sum_q \hat \chi_\delta^2(q) (\kappa^2 +
      |q|^2) |\hat u_q|^2 \right) \nonumber \\ 
      \leq \sup_q \left\{  \hat \chi_\delta^2(q) (\kappa^2 + |q|^2)
      \right\} \sum_q |\hat u_q|^2 \nonumber \\ 
      \times \int_\Omega \int_\Omega u(x) G(x - y) u(y)
      \, dx \, dy.  
  \end{eqnarray}
  Taking into account that $\sum_q |\hat u_q|^2 =
  ||u||_{L^2(\Omega)}^2 \leq M^2$ and that \cite{abramowitz}
  \begin{eqnarray}
    \label{eq:34}
    \delta^2 \hat \chi_\delta^2(q) (\kappa^2 + |q|^2) \leq 
    \begin{cases}
      C_1 (\kappa^2 m^4 M^{-2} p^{-2} + 1), & |q| \delta \leq 1, \\ 
      C_2 (\kappa^2  m^4 M^{-2} p^{-2} + |q|^2 \delta^2) (|q|
      \delta)^{-d-1}, &  |q| \delta > 1,
    \end{cases}
  \end{eqnarray}
  for some $C_{1,2} > 0$ depending only on $d$, we conclude that
  \begin{eqnarray}
    \label{eq:35}
    C \delta^2 \leq 
     \int_\Omega \int_\Omega u(x) G(x - y) u(y)
    \, dx \, dy.  
  \end{eqnarray}
  for some $C > 0$ depending only on $d$, $\kappa/p$, and $m^2/M$.
  The result then follows immediately from (\ref{eq:del2}).
\end{proof}

Let us also make some remarks regarding a few extensions of these
arguments. First, the same estimate holds true in the case where $G$
is the Green's function of the Laplacian in $\Omega$ and $u$ has zero
mean. Note that in this case the constant $C$ in (\ref{eq:17}) becomes
independent on $p$. The proof easily follows by passing to the limit
$\kappa \to 0$ in the lemma. Another observation is that, actually,
for the considered class of functions a stronger interpolation
inequality involving negative Sobolev norms holds. We give only the
statement of the result, the proof follows easily by modifying a few
steps in the arguments above

\begin{proposition}
  Let $u$ be as in Lemma \ref{l:interp}. Then
  \begin{eqnarray}
    \label{eq:41}
    \int_\Omega u \, (1 - \Delta)^{-{d + 1 \over 2}} \, u \, dx  \geq
    C  \left( \int_\Omega |\nabla u| \, dx \right)^{-d-1},
  \end{eqnarray}
  for some $C = C(d, p, m, M) > 0$.
\end{proposition}

\section{First and second variation}
\label{a:vars}

Here we present the derivation of the first and second variation of
$\bar E$ in $d = 2$, adapted from \cite{m:pre02}.

\begin{lemma}
  \label{l:vars}
  Let $\bar \Omega^+\subset \bar\Omega$ be a set with boundary of
  class $C^2$ and $v$ be given by (\ref{eq:v}). Then, the functional
  $\bar E$ is twice continuously G\^ateaux-differentiable with respect
  to $C^1$-perturbations of $\partial \bar\Omega^+$. Furthermore, the
  first and second G\^ateaux derivatives of $\bar E$ are given by
  (\ref{eq:Ebvar1}) and (\ref{eq:Ebvar2}).
\end{lemma}

\begin{proof}
  Let $a > 0$, let $\rho \in C^1(\partial \bar\Omega^+)$, and let
  $\bar\Omega^+_a$ be the set obtained from $\bar\Omega^+$ by
  transporting each point of $\partial \bar\Omega^+$ by $a \rho$ in
  the direction of the outward normal. Note that for sufficiently
  small $a$ the set $\partial\bar\Omega^+_a$ is of class $C^1$, in
  view of regularity of $\partial \bar\Omega^+$. Then, if $\bar E_a =
  \bar E(\bar\Omega^+_a)$ and $\bar E = \bar E(\bar\Omega^+)$, from
  (\ref{eq:Ebar}) we have explicitly
  \begin{eqnarray}
    \label{eq:55}
    |\ln \ep| (\bar E_a - \bar E) = \int_{\partial\bar\Omega^+}
    \left( \sqrt{(1 + a K(\bar  x) \rho(\bar x))^2 + a^2 |\nabla
        \rho(\bar x)|^2} - 1 \right) \,  d \mathcal 
    H^1(\bar x) \nonumber \\ 
    + \int_{\partial \bar\Omega^+} \int_0^{a \rho(\bar x)}
    ( 4 v(\bar x + r \nu(\bar x)) - 2 \bar\delta \kappa^{-2}) (1 +
    K(\bar x) r)  \, dr \, d  \mathcal H^1(\bar x) \nonumber \\ 
    + 2 |\ln \ep|^{-1} \int_{\partial \bar\Omega^+}
    \int_{\partial \bar\Omega^+} \int_0^{a
      \rho(\bar x)} \int_0^{a \rho(\bar
      y)}   (1 + K(\bar x) r)  (1 + K(\bar y) r')\nonumber \\ 
    \times G \bigr(\ep^{1/3} |\ln \ep|^{-1/3} (\bar x + r \nu(\bar x)
    - \bar y - r'  \nu(\bar y)) \bigl) \, dr' dr \, d  \mathcal
    H^1(\bar y) d  \mathcal H^1(\bar x),\nonumber \\
  \end{eqnarray}
  where $K(\bar x)$ is curvature, $\nu(\bar x)$ is the outward unit
  normal at $\bar x \in \partial \bar\Omega^+$, and we rewrote the
  integrals in terms of the curvilinear coordinates consisting of the
  projection $\bar x$ of a point $x \in \bar \Omega$ to
  $\partial\bar\Omega^+$ and signed distance $r = \nu(\bar x) \cdot (x
  - \bar x)$, which is possible for sufficiently small $a$. Now,
  Taylor-expanding the integrands in the powers of $r$ and integrating
  over $r$ and $r'$, after some tedious algebra we obtain that for any
  $\alpha \in (0, 1)$ it holds
  \begin{eqnarray}
    \label{eq:56}
    \bar E_a  = \bar E + a \left. {d \bar E_a \over d a}
    \right|_{a = 0} + {a^2 \over 2} \left. {d^2 \bar E_a \over d a^2}
    \right|_{a = 0} + O(a^{2+\alpha}),
  \end{eqnarray}
  where the derivatives are given by (\ref{eq:Ebvar1}) and
  (\ref{eq:Ebvar2}). In estimating the remainder term in (\ref{eq:56})
  we took into account that $v \in C^{1,\alpha}(\bar\Omega)$ and the
  following estimate of the terms involving the convolution integral: 
  \begin{eqnarray}
    \label{eq:2}
    \Biggl| \int_{\partial \bar\Omega^+} \int_0^{a \rho(\bar y)} \Bigl(
    G(\ep^{1/3} |\ln \ep|^{-1/3} (\bar x + \nu(\bar x) r - \bar y -
    \nu(\bar y) r')) \nonumber \\ 
    - G(\ep^{1/3} |\ln \ep|^{-1/3} (\bar x - \bar
    y)) \Bigr) \, dr' d \mathcal H^1(\bar y) \Biggr| \nonumber \\ 
    \leq C \int_{\partial \bar\Omega^+} \int_0^{a \rho(\bar y)}
    \Bigl(a + \left| \ln {|\bar x - \bar y + \nu(\bar x) r - \nu(\bar
        y) r'| \over |\bar x - \bar y|} \right| \Bigr) dr' d \mathcal
    H^1(\bar y)  \nonumber \\
    \leq C \Biggl( a^2 + \int_{\partial \bar\Omega^+ \cap |\bar x -
      \bar y| \geq M a} \int_0^{a \rho(\bar y)} {| \nu(\bar x) r -
      \nu(\bar y) r'| \over |\bar x - \bar y|} dr' d \mathcal H^1(\bar
    y) \Biggr) \nonumber \\
    \leq C a^2  \int_{\partial \bar\Omega^+ \cap |\bar x -
      \bar y| \geq M a} {d \mathcal H^1(\bar y) \over |\bar x - \bar
      y|} \leq C' a^2 |\ln a|, 
  \end{eqnarray}
  for $a \ll 1$, where $M > 0$ is sufficiently large, and we used the
  series expansion of $G$ \cite{abramowitz}.

  Finally, for every sufficiently small $C^1$-perturbation $\partial
  \bar\Omega^+_a$ of $\partial \bar\Omega^+$ the distance from a point
  $\bar x \in \partial \bar\Omega^+$ to $\partial \bar\Omega^+_a$ is a
  $C^1$-function, hence the formulas obtained above apply to all such
  perturbations.
\end{proof}

\bibliographystyle{springer}

\bibliography{../nonlin,../stat,../mura}

\begin{thebibliography}{10}

\bibitem{bray94}
Bray, A.J.:
\newblock Theory of phase-ordering kinetics.
\newblock Adv. Phys. \textbf{43} (1994)  357--459

\bibitem{landau8}
Landau, L.D., Lifshits, E.M.:
\newblock Course of Theoretical Physics. Volume~8.
\newblock Pergamon Press, London (1984)

\bibitem{grosberg}
Grosberg, A.Y., Khokhlov, A.R.:
\newblock Statistical Physics of Macromolecules.
\newblock AIP Press, New York (1994)

\bibitem{ko:book}
Kerner, B.S., Osipov, V.V.:
\newblock Autosolitons.
\newblock Kluwer, Dordrecht (1994)

\bibitem{vedmedenko}
Vedmedenko, E.Y.:
\newblock Competing Interactions and Pattern Formation in Nanoworld.
\newblock Wiley, Weinheim, Germany (2007)

\bibitem{muthukumar97}
Muthukumar, M., Ober, C.K., Thomas, E.L.:
\newblock Competing interactions and levels of ordering in self-organizing
  polymeric materials.
\newblock Science \textbf{277} (1997)  1225--1232

\bibitem{desimone00}
DeSimone, A., Kohn, R.V., M\"uller, S., Otto, F.:
\newblock Magnetic microstructures---a paradigm of multiscale problems.
\newblock In: ICIAM 99 (Edinburgh).
\newblock Oxford Univ. Press (2000)  175--190

\bibitem{choksi08}
Choksi, R., Conti, S., Kohn, R.V., Otto, F.:
\newblock Ground state energy scaling laws during the onset and destruction of
  the intermediate state in a type {I} superconductor.
\newblock Comm. Pure Appl. Math. \textbf{61} (2008)  595--626

\bibitem{choksi01}
Choksi, R.:
\newblock Scaling laws in microphase separation of diblock copolymers.
\newblock J. Nonlinear Sci. \textbf{11} (2001)  223--236

\bibitem{seul95}
Seul, M., Andelman, D.:
\newblock Domain shapes and patterns: the phenomenology of modulated phases.
\newblock Science \textbf{267} (1995)  476--483

\bibitem{yu07}
Yu, B., Sun, P., Chen, T., Jin, Q., Ding, D., Li, B., Shi, A.C.:
\newblock Self-assembled morphologies of diblock copolymers confined in
  nanochannels: Effects of confinement geometry.
\newblock J. Chem. Phys. \textbf{126} (2007)  204903 pp. 1--5

\bibitem{kohn07iciam}
Kohn, R.V.:
\newblock Energy-driven pattern formation.
\newblock In: International {C}ongress of {M}athematicians. {V}ol. {I}.
\newblock Eur. Math. Soc., Z\"urich (2007)  359--383

\bibitem{m:phd}
Muratov, C.B.:
\newblock Theory of domain patterns in systems with long-range interactions of
  Coulombic type.
\newblock Ph. D. Thesis, Boston University (1998)

\bibitem{m:pre02}
Muratov, C.B.:
\newblock Theory of domain patterns in systems with long-range interactions of
  {Coulomb} type.
\newblock Phys. Rev. E \textbf{66} (2002)  066108 pp. 1--25

\bibitem{care75}
Care, C.M., March, N.H.:
\newblock Electron crystallization.
\newblock Adv. Phys. \textbf{24} (1975)  101--116

\bibitem{emery93}
Emery, V.J., Kivelson, S.A.:
\newblock Frustrated electronic phase-separation and high-temperature
  superconductors.
\newblock Physica C \textbf{209} (1993)  597--621

\bibitem{chen93}
Chen, L.Q., Khachaturyan, A.G.:
\newblock Dynamics of simultaneous ordering and phase separation and effect of
  long-range {Coulomb} interactions.
\newblock Phys. Rev. Lett. \textbf{70} (1993)  1477--1480

\bibitem{nyrkova94}
Nyrkova, I.A., Khokhlov, A.R., Doi, M.:
\newblock Microdomain structures in polyelectrolyte systems: calculation of the
  phase diagrams by direct minimization of the free energy.
\newblock Macromolecules \textbf{27} (1994)  4220--4230

\bibitem{ohta86}
Ohta, T., Kawasaki, K.:
\newblock Equilibrium morphologies of block copolymer melts.
\newblock Macromolecules \textbf{19} (1986)  2621--2632

\bibitem{bates99}
Bates, F.S., Fredrickson, G.H.:
\newblock Block copolymers -- designer soft materials.
\newblock Physics Today \textbf{52} (1999)  32--38

\bibitem{matsen02}
Matsen, M.W.:
\newblock The standard {Gaussian} model for block copolymer melts.
\newblock J. Phys.: Condens. Matter \textbf{14} (2002)  R21--R47

\bibitem{degennes79}
{de Gennes}, P.G.:
\newblock Effect of cross-links on a mixture of polymers.
\newblock J. de Physique -- Lett. \textbf{40} (1979)  69--72

\bibitem{stillinger83}
Stillinger, F.H.:
\newblock Variational model for micelle structure.
\newblock J. Chem. Phys. \textbf{78} (1983)  4654--4661

\bibitem{ohta90}
Ohta, T., Ito, A., Tetsuka, A.:
\newblock Self-organization in an excitable reaction-diffusion system:
  synchronization of oscillating domains in one dimension.
\newblock Phys. Rev. A \textbf{42} (1990)  3225--3232

\bibitem{glotzer95}
Glotzer, S., Di~Marzio, E.A., Muthukumar, M.:
\newblock Reaction-controlled morphology of phase-separating mixtures.
\newblock Phys. Rev. Lett. \textbf{74} (1995)  2034--2037

\bibitem{matsen96}
Matsen, M.W., Bates, F.S.:
\newblock Unifying weak- and strong-segregation block copolymer theories.
\newblock Macromolecules \textbf{29} (1996)  1091--1098

\bibitem{choksi03}
Choksi, R., Ren, X.:
\newblock On the derivation of a density functional theory for microphase
  separation of diblock copolymers.
\newblock J. Statist. Phys. \textbf{113} (2003)  151--176

\bibitem{mnog09}
Muratov, C.B., Novaga, M., Orlandi, G., Garc\'ia-Cervera, C.J.:
\newblock Geometric strong segregation theory for compositionally asymmetric
  diblock copolymer melts.
\newblock In: Singularities in nonlinear evolution phenomena and applications.
  CRM Series.
\newblock Birkh\"auser (2009) (to appear).

\bibitem{muller93}
M\"uller, S.:
\newblock Singular perturbations as a selection criterion for periodic
  minimizing sequences.
\newblock Calc. Var. Part. Dif. \textbf{1} (1993)  169--204

\bibitem{ren00}
Ren, X.F., Wei, J.C.:
\newblock On the multiplicity of solutions of two nonlocal variational
  problems.
\newblock SIAM J. Math. Anal. \textbf{31} (2000)  909--924

\bibitem{ren03}
Ren, X., Wei, J.:
\newblock On energy minimizers of the diblock copolymer problem.
\newblock Interfaces Free Bound. \textbf{5} (2003)  193--238

\bibitem{ren07jns}
Ren, X., Wei, J.:
\newblock Single droplet pattern in the cylindrical phase of diblock copolymer
  morphology.
\newblock J. Nonlinear Sci. \textbf{17} (2007)  471--503

\bibitem{ren07rmp}
Ren, X., Wei, J.:
\newblock Many droplet pattern in the cylindrical phase of diblock copolymer
  morphology.
\newblock Rev. Math. Phys. \textbf{19} (2007)  879--921

\bibitem{ren06}
Ren, X., Wei, J.:
\newblock Droplet solutions in the diblock copolymer problem with skewed
  monomer composition.
\newblock Calc. Var. Partial Differential Equations \textbf{25} (2006)
  333--359

\bibitem{roger08}
R\"{o}ger, M., Tonegawa, Y.:
\newblock Convergence of phase-field approximations to the {Gibbs--Thomson}
  law.
\newblock Calc. Var. PDE \textbf{32} (2008)  111--136

\bibitem{alberti09}
Alberti, G., Choksi, R., Otto, F.:
\newblock Uniform energy distribution for an isoperimetric problem with
  long-range interactions.
\newblock J. Amer. Math. Soc. \textbf{22} (2009)  569--605

\bibitem{nishiura95}
Nishiura, Y., Ohnishi, I.:
\newblock Some mathematical aspects of the micro-phase separation in diblock
  copolymers.
\newblock Physica D \textbf{84} (1995)  31--39

\bibitem{choksi09}
Choksi, R., Peletier, M.A., Williams, J.F.:
\newblock On the phase diagram for microphase separation of diblock copolymers:
  an approach via a nonlocal {Cahn-Hilliard} functional.
\newblock SIAM J. Appl. Math. \textbf{69} (2008)  1712--1738

\bibitem{petrich94}
Petrich, D.M., Goldstein, R.E.:
\newblock Nonlocal contour dynamics model for chemical front motion.
\newblock Phys. Rev. Lett. \textbf{72} (1994)  1120--1123

\bibitem{goldstein96}
Goldstein, R.E., Muraki, D.J., Petrich, D.M.:
\newblock Interface proliferation and the growth of labyrinths in a
  reaction-diffusion system.
\newblock Phys. Rev. E \textbf{53} (1996)  3933--3957

\bibitem{yip06}
Yip, N.K.:
\newblock Structure of stable solutions of a one-dimensional variational
  problem.
\newblock ESAIM Control Optim. Calc. Var. \textbf{12} (2006)  721--751

\bibitem{ertl-nobel}
Ertl, G.:
\newblock Reactions at surfaces: From atoms to complexity.
\newblock
  http://nobelprize.org/nobel\_prizes/chemistry/laureates/2007/ertl-lecture.ht%
ml (2007)

\bibitem{ertl08}
Ertl, G.:
\newblock {Reactions at surfaces: From atoms to complexity (Nobel lecture)}.
\newblock Angew. Chem. Int. Ed. \textbf{{47}} ({2008})  {3524--3535}

\bibitem{wintterlin97}
Wintterlin, J., Trost, J., S, R., Schuster, R., Zambelli, T., Ertl, G.:
\newblock Real-time {STM} observations of atomic equilibrium fluctuations in an
  adsorbate system: {O/Ru(0001)}.
\newblock Surf. Sci. \textbf{394} (1997)  159--169

\bibitem{abramowitz}
Abramowitz, M., Stegun, I., eds.:
\newblock Handbook of mathematical functions.
\newblock National Bureau of Standards (1964)

\bibitem{theil06}
Theil, F.:
\newblock A proof of crystallization in two dimensions.
\newblock Comm. Math. Phys. \textbf{262} (2006)  209--236

\bibitem{aftalion07}
Aftalion, A., Serfaty, S.:
\newblock Lowest {L}andau level approach in superconductivity for the
  {A}brikosov lattice close to {$H\sb {c\sb 2}$}.
\newblock Selecta Math. \textbf{13} (2007)  183--202

\bibitem{chen07arma}
Chen, X., Oshita, Y.:
\newblock An application of the modular function in nonlocal variational
  problems.
\newblock Arch. Ration. Mech. Anal. \textbf{186} (2007)  109--132

\bibitem{massari74}
Massari, U.:
\newblock Esistenza e regolarit\`a delle ipersuperfice di curvatura media
  assegnata in {$R\sp{n}$}.
\newblock Arch. Rational Mech. Anal. \textbf{55} (1974)  357--382

\bibitem{gilbarg}
Gilbarg, D., Trudinger, N.S.:
\newblock Elliptic Partial Differential Equations of Second Order.
\newblock Springer-Verlag, Berlin (1983)

\bibitem{giusti}
Giusti, E.:
\newblock Minimal surfaces and functions of bounded variation. Volume~80 of
  Monographs in Mathematics.
\newblock Birkh\"auser, Basel (1984)

\bibitem{choksi07}
Choksi, R., Sternberg, P.:
\newblock On the first and second variations of a nonlocal isoperimetric
  problem.
\newblock J. Reine Angew. Math. \textbf{611} (2007)  75--108

\bibitem{ko:mk85}
Kerner, B.S., Osipov, V.V.:
\newblock Phenomena in active distributed systems.
\newblock Mikroelektronika \textbf{14} (1985)  389--407

\bibitem{m1:prl97}
Muratov, C.B.:
\newblock Instabilities and disorder of the domain patterns in the systems with
  competing interactions.
\newblock Phys. Rev. Lett. \textbf{78} (1997)  3149--3152

\bibitem{mo1:pre96}
Muratov, C.B., Osipov, V.V.:
\newblock General theory of instabilities for pattern with sharp interfaces in
  reaction-diffusion systems.
\newblock Phys. Rev. E \textbf{53} (1996)  3101--3116

\bibitem{fusco08}
Fusco, N., Maggi, F., Pratelli, A.:
\newblock The sharp quantitative isoperimetric inequality.
\newblock Ann. of Math. \textbf{168} (2008)  941--980

\bibitem{m1:pre97}
Muratov, C.B.:
\newblock Synchronization, chaos, and the breakdown of the collective domain
  oscillations in reaction-diffusion systems.
\newblock Phys. Rev. E \textbf{55} (1997)  1463--1477

\bibitem{kiessling99}
Kiessling, M.K.H., Spohn, H.:
\newblock A note on the eigenvalue density of random matrices.
\newblock Comm. Math. Phys. \textbf{199} (1999)  683--695

\bibitem{struwe}
Struwe, M.:
\newblock Variational methods : applications to nonlinear partial differential
  equations and Hamiltonian systems.
\newblock Springer, Berlin (2000)

\bibitem{ko:jetp80}
Kerner, B.S., Osipov, V.V.:
\newblock Stochastically inhomogeneous structures in nonequilibrium systems.
\newblock Sov. Phys. -- JETP \textbf{52} (1980)  1122--1132

\bibitem{mimura80}
Mimura, M., Tabata, M., Hosono, Y.:
\newblock Multiple solutions of two-point boundary value problems of {Neumann}
  type with a small parameter.
\newblock SIAM J. Math. Anal. \textbf{11} (1980)  613--631

\bibitem{modica77}
Modica, L., Mortola, S.:
\newblock Un esempio di {$\Gamma \sp{-}$}-convergenza.
\newblock Boll. Un. Mat. Ital. B \textbf{14} (1977)  285--299

\bibitem{modica87}
Modica, L.:
\newblock The gradient theory of phase transitions and the minimal interface
  criterion.
\newblock Arch. Rational Mech. Anal. \textbf{98} (1987)  123--142

\bibitem{attouch}
Attouch, H., Buttazzo, G., Michaille, G.:
\newblock Variational analysis in Sobolev and BV spaces.
\newblock Society for Industrial and Applied Mathematics, Philadelphia (2006)

\bibitem{fife77}
Fife, P.C., McLeod, J.B.:
\newblock The approach of solutions of nonlinear diffusion equations to
  traveling front solutions.
\newblock Arch. Rat. Mech. Anal. \textbf{65} (1977)  335--361

\end{thebibliography}

\end{document}